\documentclass[10 pt]{article}
\usepackage[utf8]{inputenc}
\usepackage[fleqn]{amsmath}
\usepackage{array}
\usepackage{appendix}
\usepackage{romannum}
\usepackage{amsfonts}
\usepackage{enumitem}
\usepackage{soul}
\usepackage{amssymb,latexsym}
\usepackage{graphics}
\usepackage{float}
\usepackage{subcaption}
\usepackage{graphicx}
\usepackage{comment}
\usepackage{mathtools}
\usepackage{amsfonts}
\usepackage{amsthm}
\newtheorem{theorem}{Theorem}

\usepackage{multicol}
\usepackage{cite}
\usepackage{tikz}
\usepackage{amsmath, amsthm, amscd, amsfonts, amssymb, graphicx}
\usepackage[left=1.5 cm,right=1.5 cm,top=1.5cm,bottom=1.5cm]{geometry}
\title{Two predators one prey model that integrates the effect of supplementary food resources due to one predator's kleptoparasitism under the possibility of retribution by the other predator}
\author{Debasish Bhattacharjee$^1$,Dipam Das$^2$, Santanu Acharjee$^3$,Tarini Kumar Dutta$^4$\\
$^{1,2,3}$Department of mathematics,Gauhati University,Assam,India\\
$^4$Department of Mathematics, Assam Don Bosco University, Assam, India\\
e-mails:$^{1}$debabh2@gmail.com,
$^{2}$pontu.dd@gmail.com,
$^{3}$sacharjee326@gmail.com
$^{4}$tkdutta2001@yahoo.co.in }
\date{}
\begin{document}

\maketitle
\begin{abstract}

In ecology, foraging requires animals to expend energy in order to obtain resources. The cost of foraging can be reduced through kleptoparasitism, the theft of a resource that another individual has expended effort to acquire. Thus, kleptoparasitism is one of the most significant feeding techniques in ecology. In this study, we investigate a two predator one prey paradigm in which one predator acts as a kleptoparasite and the other as a host. This research considers the post-kleptoparasitism scenario, which has received little attention in the literature. Parametric requirements for the existence as well as local and global stability of biologically viable equilibria have been proposed. The occurrences of various one parametric bifurcations, such as saddle-node bifurcation, transcritical bifurcation, and Hopf bifurcation, as well as two parametric bifurcations, such as Bautin bifurcation, are explored in depth.  Relatively low growth rate of first predator induces a subcritical Hopf bifurcation although a supercritical Hopf bifurcation occurs at relatively high growth rate of first predator making coexistence of all three species possible. Some numerical simulations have been provided for the purpose of verifying our theoretical conclusions.

\end{abstract}

\textbf{2020 AMS Classifications:} 92D40, 92D25, 34D20, 34C23\\


\section{Introduction}
\pagenumbering{arabic}
 Predator-prey relationships are a pervasive and diverse natural phenomenon that exerts a significant influence on population dynamics. As such, ecologists have devoted considerable time and resources to their study, with mathematical modelling being a valuable tool for understanding the underlying rationale and strategies that govern these interactions. The study of the dynamics of predator-prey systems dates back to the works of Lotka \cite{lotka} and Volterra \cite{Vol}. The Lotka-Volterra model, which describes the interaction between two species using a set of non-linear coupled first-order ordinary differential equations, is the fundamental model. 
 In order to establish the validity of modelling expressions in comparison to reality, numerous modifications have been offered since then to account for a range of ecological phenomena. Despite the fact that numerous ecological models have been developed and have significantly aided our understanding of how prey and predator interactions operate within a system, numerous ecological systems are emerging that do not fall within the collection of mathematical transformations of prey-predator ecological systems. Starting with a simple one-prey, one-predator interaction environment,  mathematical models are evolving to incorporate various ecological complexities such as the allee effect \cite{Moro, SenD, SahaS, Ali}, fear effect \cite{magh, Jyo, Yang}, memory effect \cite{Ghos, Ali, Yang} and others, which exist even in ostensibly simple ecosystems and are a recent area of interest.

 \par The addition of additional species or population divisions to the simple prey-predator system framework and the quantification of the qualitative complexity of the ecosystem constitutes a second avenue for the advancement of mathematical modelling that is consistent with reality. Numerous academicians have devised mathematical models depicting interactions between two predators and a single prey\cite{Long,Gao,Dji,Ducr,Xu,muk,kar,Dia,Har,Sav,Jau}, two prey and a single predator\cite{magh,Saho}, two predators and two prey species\cite{Abhi,Farh}, etc. Such expansions might offer insight on the intricacy of species relationships outside of the usual prey-predator conflict, in addition to outlining other ecological complications. In this endeavor, we incorporate a model with two predators and one prey,  enriched with complex interactions of prey-predator and predator-predator.\\
 \par Prey-predator interaction and predator-predator interaction are the principal kinds of direct interactions that transpire between populations in ecological systems when two predator species coexist with the same prey species, and the heart of the dynamics of the ecosystem is the prey-predator interaction or the functional response. 
 The term "functional response of the predator" refers to the rate at which a predator consumes its prey. It is one of the important prey-predator interactions which is responsible for death of both prey and sometimes predator populations. This rate is measured by the average number of prey that is killed by each predator per unit of time. Different types of functional responses have been documented within the literature of theoretical ecology. Holling \cite{holl} described three distinct functional responses,  Holling Type-I, Holling Type-II, and Holling Type-III, to explain the effects of various forms of predation.There are numerous other types of functional responses, each with its own distinctive properties\cite{Hale,Mark,Free,Har,Hol,holl}. In this paper, we employ the Holling type-I functional response, which predicts a linear increase, i.e., the time required for a predator to handle a food item is negligible or that eating doesn't get in the way of foraging for more food. Prey-predator models with a variety of functional responses have been investigated by several researchers \cite{Dia,SahaS,Didi}.\\

  \par When two predator species subsist off of the same prey, it is impossible to disregard their competition for resources, as competition for resources is common in nature and society \cite{Jau}. Competition constitutes one of the most important predator-prey interactions. In conformance with this concept, numerous researchers \cite{Dia,muk,dub,Sav} have considered a three-species predator-prey system in which the two predators compete for the same prey. Our model accounts for the interspecific animosity between the two predator species.\\
  
 \par  This research investigates another type of predator-predator interaction found in nature: 
 "kleptoparasitism" . Although kleptoparasitism is one of the most prevalent foraging strategies in nature for procuring scarce resources and many animals rely heavily on it for subsistence, it has received little consideration in the construction and analysis of predator-prey population models. Although kleptoparasitism has been implicitly incorporated into a small number of mathematical models of ecosystems. For instance, many ecologists \cite{Ruxo,Ross,Spen,Hadj} discussed a game-theoretical model of kleptoparasitism. Broom et al.\cite{Mark,Yate} discussed a stochastic model of kleptoparasitism. But in these models, ecologists have focused only in the kleptoparasitism scenario, but the whole population dynamics in long term is not shown. Focardi et al \cite{Stef} investigated the role of scavenging in predator and multiprey dynamics in anorthern Apennine system in Italy. Materassi et al. \cite{Mass} are the first to explicitly introduce kleptoparasitism in trophic web models capable of explaining long-term population dynamics. Other than these two works of literature, the authors are not aware of any other predator-prey models that mathematically explore the kleptoparasitism effect in long-term population dynamics.\\

 \par Kleptoparasitism is the feeding strategy where one animal takes resources from another by means of stealing. Rothschild and Clay \cite{ROT} coined the word "clepto-parasitism" or "kleptoparasitism" to refer the act of one species' members stealing food from the other's members of a different species. The same behaviour is also referred to as ‘piracy’ by many authors \cite{MEI,ash} as well as ‘food parasitism’ \cite{hop}, and ‘pilfering’  \cite{ran}. This feeding strategy becomes more important when food is scarce in nature or it takes too much energy to hunt.  Kleptoparasites obtain food without actually expending time and energy searching for and capturing prey. Kleptoparasitism allows the shift of investment from time and energy from foraging to other activities, such as searching for mates or avoidance of predators \cite{Iyen}. Thus, kleptoparasitism boosts the growth rate of that species. Many examples of kleptoparasitism are found in nature. For example, kleptoparasitism in ant-eating spiders of the genus Zodarion. In the field study, it is found that adult females consistently hunted actively, while adult males ceased active prey capture and instead engaged in kleptoparasitism heavily to maximize mating opportunities\cite{Mart}.Kleptoparasitism is widespread among birds and is an important feeding technique of some families of seabirds, notably Fregatidae (frigatebirds), Chionididae (sheathbills), Stercoraridae (skuas), and Laridae (gulls and terns). For instance Skuas and frigatebirds, obtain a large proportion of their food by kleptoparasitism of heterospecifics\cite{Hock}.  
Spotted hyenas (Crocuta crocuta) and lions (Panthera leo) are potentially serious competitors for food, as they show a significant overlap in their diet \cite{Kru,Mil}. Both species are subject to kleptoparasitism from each other \cite{HON,hon,per}. We address the effects of kleptoparasitism on two predators and one prey in our model, which further incorporates the post-kleptoparasitism repercussions in terms of energy loss from the predator's counter-attack, despite the fact that this aggression is not fatal. Transformation of the immensely natural post-kleptoparasitism effect into a mathematical form and incorporation into a mathematical model to demonstrate its long-term influences in population dynamics is coined for the first time so far as the author's knowledge.   \\

\par Risks are associated with kleptoparasitism. The reaction of the food owner towards a kleptoparasite differs depending on its dominance status. In many cases, it can be seen that the food owner engages in a counter-attack on the Kleptoparasite\cite{Hans}. We can find evidence of counter-attack during kleptoparasitism in nature. Black-headed gulls, Larus ridibundus, rely on surprise and agility to steal prey from a more dominant species. It is observed that black-headed gulls( Larus ridibundus) successfully grab in a flight  worm from the bill of a curlew(Numenius arquata), while withdrawing in time to escape the counter-attack of the startled victim\cite{Ens}.  Experimentally, Rooks(Corvus frugilegus) are found to be especially susceptible to kleptoparasitism. Although, sometimes Rooks make counter-attacks on the kleptoparasite\cite{Hans}. The kleptoparasite can mitigate the impact of these counter-attacks by engaging more individuals in kleptoparasitism because the effects of the reprisal can be distributed among the kleptoparasites.  \\

\par In endorsement of the aforementioned natural phenomena, we delve into a mathematical model of two predators vying for the same prey while one predator species engages in kleptoparasitism with the possibility of a counter-attack. In Section 2, the mathematical modelling is presented along with an explanation of its components. Section 3 covers areas such as the solution's positivity and boundedness, local stability as well as global stability  analysis.Existence of different types of bifurcations, direction and stability of Hopf bifurcations are demonstrated in section 4. In section 5, we validate the results through various numerical simulations, and in section 6, we derive the necessary conclusions.

\section{\textit{Model construction}} Within this section, a mathematical model is developed based on our observations of the phenomenon of kleptoparsitism among predator species. Prior to formulating the model, the following assumptions are established:
\begin{enumerate}
\item We consider a scenario wherein two predator species vie for the same prey. At any given time $t$, $s$ symbolizes the population size of the prey species, while $p_1$ and $p_2$ represent the population sizes of the first and second predator species, respectively.\\
\item We assume that prey growth is logistic in nature in the absence of the predator, owing to limited resources in nature. The carrying capacity of the ecosystem is considered to be $k$.\\
\item We make the
 assumption that the first predator(of size $p_1$) eats up the prey(of size $s$) in accordance with the Holling Type-I functional response. The second predator utilises the same Holling Type-I functional response to consume its prey as the first. We hypothesise that the prey is very small in comparison to both predator species, as a result, handling time is minimal, and it therefore qualifies for the Holling type I functional response.
 
\item Competition arises when the presence of other individuals reduces one's ability to obtain resources. Predator competition can arise when members of the same or different species compete for the same prey. 
We assume that these two predators are engaged in interspecific competition with each other.

    
    \item We assume a natural death in either of the two predator species.  The aforementioned assumptions lead to the following model equations:
   \begin{equation} \label{a}
       \begin{split}
           &\frac{ds}{dt}=r_1s(1-\frac{s}{k})-sp_1r_2-sp_2r_3 \\
           &\frac{dp_1}{dt}=r_4r_2sp_1-k_1p_1p_2-d_2p_1 \\
           &\frac{dp_2}{dt}=r_5sp_2r_3-k_2p_1p_2-d_3p_2
       \end{split}
   \end{equation}

This mathematical model is studied by Hsu et al. \cite{Hsu}. The authors studied the relationship between the  coefficient of interference and competition outcome.

 \item Addition to (\ref{a}), we consider that second predator involves in kleptoparsitism and steals food or resources from the first predator and apart from that second predator also hunts on its prey(of size $s$). So, in this study, the second predator is regarded as the kleptoparasite and the first predator as the host. In kleptoparasitism behaviour, the predator that is the victim of kleptoparasitism(host) faces a heavy loss in their growth rate, while the predator that performs kleptoparasitism(kleptoparasite) enjoys a positive effect on their growth.
 Let  $\psi=\frac{1}{1+ap_2}$, where 
    $\psi$ denotes the fraction of s killed by per capita of first predator and really eaten by them,  i.e.converted into their biomass. So, $\psi \in$ (0,1) and (1-$\psi$) will denote the fraction of s killed by per capita of first predator species but not eaten by them. Kleptoparasitism by the second predator costs per capita of first predator a fraction $(1-\psi)$ of their food. Here \emph{a} denotes the intensity of kleptoparasitism. Evidently, the following ecological hypotheses are satisfied by the function $\psi$:\\
    (\romannum{1}) $ \lim_{a \to 0} \psi =1$ ; i.e., if kleptoparsitism intensity of the second predator is negligible then the first predator eats all their kills.    \\
(\romannum{2}) $\lim_{a \to \infty} \psi =0$ ; i.e., every kill of the first predator will be plundered if the second predator's rate of kleptoparasitism rises significantly, leaving nothing for the first predator to imbibe.
 \\
Total food for the first predator by hunting is = $r_2 s p_1$\\
Incorporating kleptoparasitisim by the second predator, the amount of food for the first predator will be reduced to $\frac{r_2 s p_1}{1+ap_2}$, which will be converted to their biomass with conversion rate $r_4$.
\item  The amount of food available for the second predator to consume through kleptoparasitism is $\frac{s p_1 p_2 r_2 a}{1+ap_2}$. Let the conversion rate of the food acquired through kleptoparasitism by the second predator species be $r_7$. The second predator is considered to be capable of kleptoparasitism either individually or collectively, whereas the first predator pursues independently. We assume that $r_7$ is dependent on the population size of the second predator. This is a consequence of the post-kleptoparasitism effect, which is the first predator's counter-attack or apprehension of one. When the second predator engages in kleptoparasitism over the first predator, there is a possibility of a counter-attack by them due to rage. It is obvious that the second predator can provide a better group defence for their stolen food from the first predator as its size increases. Hence , the remaining amount of  stolen food after counter-attack is dependent directly to the population size of the second predator. So we may assume that the conversion rate of the stolen food is dependent on the population size of the second predator.\\

 There is a possibility that as the second predator's group increases, the concern of a first predator's reprisal diminishes, which could lead to an increase in the conversion rate. According to the aforementioned situation, it is reasonable to infer that the population size of the second predator affects how much of the stolen food is converted.
We assume $r_7 =r_6^{'} (1+ap_2)$, where $r_6^{'}$ is proportionality constant and 
$ \lim_{p_2\to 0} r_7 =r_6^{'}$.
 Therefore, the amount $\frac{s p_1 p_2 r_2 a}{1+ap_2}$ will get converted into the biomass of the second predator, and the term will convert into
    $\frac{r_7 s p_1 p_2 r_2 a}{1+ap_2}= r_6^{'} a r_2 s p_1 p_2=r_6sp_1p_2$
  , where $r_6$=$r_6^{'}ar_2$.\\
 \item If the first predator is successful in driving the second predator away from the food by mounting a counter-attack after being kleptoparasitized by the second predator, the stolen food which is recovered from the second predator is assumed not to contribute to its rate of development.

     \end{enumerate}

     Under these new circumstances, the system (\ref{a}) changes into the following modified version:\\
    \begin{equation} \label{b}
        \begin{split}
            &\frac{ds}{dt}=r_1s(1-\frac{s}{k})-sp_1r_2-sp_2r_3 \\
            &\frac{dp_1}{dt}=\frac{r_4r_2sp_1}{(1+ap_2)}-k_1p_1p_2-d_2p_1\\
            &\frac{dp_2}{dt}=r_5sp_2r_3-k_2p_1p_2-d_3p_2 +r_6p_1p_2s
        \end{split}
\end{equation}
with initial conditions: $s(0)= s^0>0 ,p_1(0) = p_1^0 >0$ and $p_2(0)=p_2^0>0$\\

Here, '$r_1$' represents the intrinsic growth rate of the prey species, '$r_2$' and '$r_3$' represent the first and second predator's predation rates, respectively. '$r_4$' and '$r_5$' represent the conversion efficiencies of the first and second predators, respectively, in the event of direct predation. '$k_1$' and '$k_2$' represent interspecific competition rates for the first predator and second predator respectively, '$d_2$' and '$d_3$' represents death rates of the first and second predator, respectively. '$k$' is the environmental carrying capacity of prey, and '$a$' is the intensity of kleptoparasitism. \\

For further simplification of the model, we contemplate the following transformations
 $$ s=Sk , p_1= \frac{P_1r_1}{r_2} , p_2=\frac{P_2r_1}{r_3},  T=r_1 t $$  and the  biosystem (\ref{b}) reduces to the following system (using t instead of T for notational convenience)
\begin{equation} \label{c}
    \begin{split}
       & \frac{dS}{dt}=S(1-S)-SP_1-SP_2\\
        & \frac{dP_1}{dt}=\frac{A_4SP_1}{(A_1+P_2)}-A_2P_1P_2-A_3P_1\\
        & \frac{dP_2}{dt}=A_5SP_2-P_1P_2A_6-A_7P_2+A_8P_1P_2S \\
    \end{split}
\end{equation}
With initial conditions $S(0)=S^0, P_1(0)=P_1^0$ and $P_2(0)=P_2^0$ .\\
Here $ A_1= \frac{r_3}{ar_1}$ , $A_2=\frac{k_1}{r_3}$ , $A_3=\frac{d_2}{r_1}$, $A_4=\frac{r_2r_4kr_3}{ar_1^2}$, 
$A_5=\frac{r_5kr_3}{r_1}$, $A_6=\frac{k_2}{r_2}$, $A_7=\frac{d_3}{r_1}$, and $A_8=\frac{r_6k}{r_2}$.
\\


\section{\textit{Ecological feasibility and sustainability of the model}}
The survival of species is an integral component of an ecosystem and is dependent on it directly or indirectly. Consequently, the species cannot expand indefinitely with finite resources. As a result, a model must satisfy both the positivity and boundedness of a species in order to accurately represent a particular ecosystem. In this section, we evaluate positivity and boundedness to determine ecological feasibility, followed by the equilibrium points of the proposed system (3) and their stability to ensure its long-term viability.

\subsection{\textit{Positivity and Boundedness }} 
Since the solution of the system (\ref{c}) indicates populations, thus it is crucial to demonstrate that it is both positive and bounded. To begin proving that our current system (\ref{c}) has only positive solutions, we may express it as follows:
\begin{equation} \label{d}
\begin{split}
    &\frac{dS}{dt}=S\phi_1(S,P_1,P_2) \\
    &\frac{dP_1}{dt}=P_1\phi_2(S,P_1,P_2) \\
    &\frac{dP_2}{dt}=P_2\phi_3(S,P_1,P_2) \\
\end{split}
    \end{equation} 
    where,  $\phi_1(S,P_1,P_2)=(1-S)-P_1-P_2$ ,  $\phi_2(S,P_1,P_2)=\frac{A_4S}{(A_1+P_2)}-A_2P_2-A_3$ , $\phi_3(S,P_1,P_2)=A_5S-P_1A_6-A_7+A_8P_1S$.
\\

   We now apply the following theorem to demonstrate that all the solutions of system system (\ref{c}) are positive.
   
\begin{theorem}\label{t1}
\cite{kar} Any solution of the differential equation $\frac{dX}{dt}=X \xi(X,Y)$ is always positive.
\end{theorem}


\begin{theorem}
Solutions to system (\ref{c}) are always positive.
\end{theorem}
\begin{proof}
 Theorem \ref{t1} gives proof for the theorem since system (\ref{c}) is stated as  system (\ref{d}).
\end{proof}
\begin{theorem}
   All the solutions of the system (\ref{c}) are bounded. 
\end{theorem}
\begin{proof}
     The first equation of the system that represents the prey equation is bounded through $\frac{dS}{dt}\le S(1-S) $ . By solving the differential
inequality we get,  $\limsup\limits_{t \rightarrow \infty} S(t) \le 1$ or $S(t) \le  1$.\\
Now, we define a function, $D(t)=S+\frac{P_1}{A_9} +\frac{P_2}{A_5}$ , and then by taking the derivative along
the solution of system (2),we get
\small$$\frac{dD}{dt}\le S(1-S)-P_1P_2(\frac{A_2}{A_9} +\frac{A_{10}}{A_5})-\frac{A_3}{A_9}P_1-\frac{A_7}{A_5}P_2$$\normalsize
Now, we choose $\phi$ in such a way that $ \phi < min ({A_3, A_7}) $, so that the above inequality becomes
\small$$\frac{dD}{dt}+ \phi D \le S(1-S+ \phi )-P_1P_2(\frac{A_2}{A_9} +\frac{A_{10}}{A_5})-(\frac{A_3}{A_9} -\frac{\phi}{A_9} ) P_1-(\frac{A_7}{A_5}- \frac{\phi}{A_5} )  P_2$$
which reduces to
$$\frac{dD}{dt}+ \phi D \le S(1-S+ \phi ) \le \frac{(1+ \phi)^2}{4} = \eta $$
 Here, $A_9=\frac{A_4}{A_1}$ and $A_{10}=(A_6-A_8)$ . Now, by using theory of differential inequality, we get 
 $ 0 < D(t)\le \frac{\eta}{\phi} (1-e^{-\phi t}) +D(0)e^{-\phi t}$\normalsize  
  and for large
value of t, i.e., t → $\infty$  , we have $0 < D(t) \le$  $\frac{\eta}{\phi}$ that is independent of initial conditions and 
hence, the system (\ref{c}) is
bounded. Thus all the solutions of the system (\ref{c}) are confined in the region $\Phi =  \left((S,P_1,P_2): 0 \le S+\frac{P_1}{A_9} +\frac{P_2}{A_5} \le \frac{\eta}{\phi} + \theta, \forall \theta > 0\right)$.

\end{proof}

\subsection{\textit{Existence of  equilibrium points}}
This section discusses the equilibrium points of the system (\ref{c}) and their existence requirements. There are six distinct equilibrium points that are mathematically defined, but only five of them are biologically significant. They are as follows:\\
\begin{enumerate}[label=(\alph*)]
    \item Extinct equilibrium $E_0$ $\equiv$ (0,0,0) always exists without any parametric conditions. It depicts an ecosystem devoid of any species habitats under investigation. 
 \\
    
    \item  Predator free equilibrium  $E_1 \equiv$ (1,0,0) always exists and it represents the situation of only prey's survival up to the full carrying capacity.\\
    
    \item $E_2 \equiv
    (B_1,C_1,0)$ represents one of the boundary equilibrium points which is free from the second predator. Here, $B_1=\frac{A_1A_3}{A_4}$ , $C_1= \frac{-A_1A_3+A_4}{A_4}$ and it exists under the condition $A_4 >A_1A_3$. It is the hypothetical situation that explains the eradication of the second predator from the ecosystem.\\
    
    \item  $E_3 \equiv (B_2,0,D_1)$ represents one of the boundary equilibrium points which is free from the first predator. Here, $B_2=\frac{A_7}{A_5}$ , $D_1=\frac{A_5-A_7}{A_5}$ and it exists under the condition $A_5>A_7$. It is the hypothetical situation that explains the eradication of the first predator from the ecosystem.\\
    
    \item The equilibrium point $E_4 \equiv (0,C_2,D_2)$ portrays the prey free scenario in the ecosystem. Here, $C_2= -\frac{A_7}{A_6}$, $D_2=- \frac{A_3}{A_2}$ and obviously it does not exist. 
   
     \item  $E^* \equiv (S^{*},P_1^*,P_2^*)$ represents interior equilibrium point or coexisting equilibrium point in which all the three species coexist.
       
\end{enumerate}

$E^*$ is a positive real solution of the system of equations
 \begin{equation} \label{j}
    \begin{split}
       & (1-S)-P_1-P_2=0\\
       & \frac{A_4S}{(A_1+P_2)}-A_2P_2-A_3=0\\
       & A_5S-P_1A_6-A_7+A_8P_1S =0\\
    \end{split}
\end{equation}

Solving prey nullcline and first predator nullcline for $P_1$ and $P_2$ and substituting these values in the second nullcline,it is found that S satisfies a polynomial equation of degree 4 which is given by \\
\begin{equation} \label{e}
\begin{split}
   \gamma_1 S^{*4}+ \gamma_2 S^{*3}+\gamma_3 S^{*2}+\gamma_4 S^{*}+\gamma_5=0
\end{split}
\end{equation}

Let,
$S^{*}$ be a positive root of this above polynomial, then \small$P_1^*= \frac{1}{2} \left(-\frac{\sqrt{A_1^2 A_2^2-2 A_1 A_2 A_3+4 A_2 A_4 S^{*}+A_3^2}}{A_2}+A_1+\frac{A_3}{A_2}-2 S^{*}+2\right) $ ,  provided  $A_4<\frac{A_2+A_1A_2+A_3+A_1A_3-2A_2S^{*}-A_1A_2S^{*}-A_3S^{*}+A_2S^{*2}}{S^{*}}$    and $P_2^*= \frac{\sqrt{A_1^2 A_2^2-2 A_1 A_2 A_3+4 A_2 A_4 S^{*}+A_3^2}-A_1 A_2-A_3}{2 A_2}$, provided $A_4>\frac{A_1A_3}{S^{*}}.$ \normalsize\\

Here, $\gamma_1=4 A_2^2A_8^2$ , 
 $ \gamma_2=-4 A_2 A_8 (2 A_2 (A_5 + A_6) + ((2 + A_1) A_2 + A_3 + A_4) A_8)$ , $\gamma_3= 4 A_2 (A_2 (A_5 + A_6)^2 + (A_3 A_5 + 
      2 (A_3 + A_4) A_6 + (2 + A_1) A_2 (A_5 + 2 A_6) + 2 A_2 A_7) A_8 + (1 + 
      A_1) (A_2 + A_3) A_8^2)$ , $\gamma_4= -4 A_2 (A_4 A_6^2 + A_3 A_6 (A_5 + A_6) + A_2 (A_5 + A_6) ((2 + A_1) A_6 + 2 A_7) +
    A_3 (2 (1 + A_1) A_6 + A_7) A_8 + A_2 (2 (1 + A_1) A_6 + (2 + A_1) A_7) A_8)$, and  $\gamma_5=4 A_2 (A_6 + A_1 A_6 + A_7) (A_3 A_6 + A_2 (A_6 + A_7)) $. \\

\textbf{Remark 3.2.1.} \emph{In section 5 , rigorous numerical analysis is used to demonstrate the existence of coexisting equilibrium point.}

\subsection{\textit{Local stability analysis }} The linear approximation of the system (\ref{c}) is considered around each equilibrium state in order to analyse the local behaviour of the system (\ref{c}). The variational matrix of the linearized system at the point $(S,P_1,P_2)$ is denoted by

$$J(S,P_1,P_2)= \left[\begin{smallmatrix}
1-2S-P_1-P_2 & -S & -S \\
\frac{A_4P_1}{A_1+P_2} & \frac{A_4S}{A_1+P_2} -A_2P_2-A_3 & - \frac{A_4SP_1}{(A_1+P_2)^2} -A_2P_1 \\
A_5P_2 +A_8P_1P_2 & -P_2A_6+A_8P_2S & A_5S -P_1A_6-A_7+A_8P_1S 
\end{smallmatrix}\right] 
$$ \\

Let us assume,the characteristic equation of $J(S,P_1,P_2)$ at $E^*$ is given by

$$\lambda^3 + G_1 \lambda^2 + G_2\lambda + G_3 = 0$$

\begin{theorem}\label{t3}
    The nature of the equilibrium points of the system (\ref{c}) are as follows:\\
(\Romannum{1})\textit{$E_0$  is always unstable.}\\
(\Romannum{2})\textit{ $E_1$ is locally asymptotically stable if the following criteria hold:\\
(\romannum{1}) $A_1A_3 > A_4$ and (\romannum{2}) $A_7>A_5$ hold.\\
(\Romannum{3})  The predator-free equilibrium point $E_2$ is locally asymptotically stable if certain requirements are satisfied: \\
\small (\romannum{1}) $A_1A_3 < A_4 $, (\romannum{2}) $A_5\le \frac{A_4 A_6-A_1 A_3 A_6}{A_1 A_3}$ and (\romannum{3}) $A_8< \frac{A_1 A_3 A_4 A_5+A_1A_3 A_4 A_6-A_4^2 A_6-A_4^2 A_7}{A_1^2 A_3^2 -A_1 A_3 A_4}$\normalsize
 hold.} \\
(\Romannum{4}) \textit{$E_3$ is locally asymptotically stable if conditions \\
(\romannum{1}) $ A_5> A_7$ and (\romannum{2}) $A_4\le \frac{A_2 A_5^2-2A_2A_5A_7+A_2A_7^2+A_3A_5^2-A_3A_5A_7}{A_5 A_7}$
hold.} \\
(\Romannum{5})\textit{Local asymptotic stability of the coexistence equilibrium state $E^*$ holds if and only if $ G_1, G_2, G_3$ and $G_1G_2-G_3$ are positive, where $G_1,G_2$ and $G_3$ have usual meanings.}
\end{theorem}

\begin{proof}
    \romannum{1}) The variational matrix in the neighbourhood of the trivial equilibrium point $E_0$ is\\
$$J(E_0)= \begin{bmatrix}
1 & 0 & 0 \\
0 & -A_3 & 0 \\
0 & 0 & -A_7
\end{bmatrix}
$$ \\
The variational matrix $J(E_0)$ has eigenvalues 1, -$A_3$ , and $-A_7$. One eigenvalue is positive, while the other two are negative. Thus, the predator-prey system (\ref{c}) is unstable near trivial equilibrium point $E_0$.\\

\romannum{2}) The Jacobian matrix around $E_1$ is\\
$$J(E_1)= \begin{bmatrix}
-1 & -1 & -1 \\
0 & \frac{A_4 -A_1A_3}{A_1} & 0 \\
0 & 0 & A_5-A_7
\end{bmatrix}
$$ \\
Now, J($E_1$) has eigenvalues  -1, $ \frac{A_4 -A_1A_3}{A_1}$ , and $A_5-A_7$. 
Now, all the three eigenvalues become negative  if $A_1A_3 > A_4$, and $A_7>A_5$ hold. As a result, $E_1$ is locally asymptotically stable if the aforementioned two conditions are met simultaneously.  $E_1$ losses its stability if either $A_1A_3 < A_4$ or $A_7 < A_5$ or both hold.\\

\romannum{3})The Jacobian matrix around $E_2$ is \\
$$ J(E_2)= \begin{bmatrix}
a_{11} & a_{12} & a_{13} \\
a_{21} & a_{22}& a_{23} \\
a_{31} & a_{32} & a_{33}
\end{bmatrix}
$$ 
where, $a_{11}= 1-2B_1-c_1$ ,  $a_{12}= -B_1$ , $a_{13}=-B_1,$ 
$a_{21}= \frac{A_4C_1}{A_1}$, $a_{22}= 0$, $a_{23}= -\frac{A_4B_1C_1}{A_1^2} -A_2C_1$, and 
$a_{31}=0$, $a_{32}=0$, $a_{33}=A_5B_1-C_1A_6-A_7+A_8B_1C_1$. \\
The characteristic equation of  J($E_2$) is 
$$\kappa^3 + F_1 \kappa^2 + F_2 \kappa + F_3 = 0,$$\\
where , $F_1=-(a_{11}+a_{33})$ , $F_2=a_{11} a_{33}-a_{11}a_{21}$ and $F_3=a_{12}a_{21}a_{33}$. \\

Taking,
$A_1A_3 < A_4 $ , 
$A_5> \frac{A_4 A_6-A_1 A_3 A_6}{A_1 A_3}$ , 
$A_8< \frac{A_1 A_3 A_4 A_5+A_1A_3 A_4 A_6-A_4^2 A_6-A_4^2 A_7}{A_1^2 A_3^2 -A_1 A_3 A_4}$ , and $A_7 > \frac{(A_1 A_3 A_5 + A_1 A_3 A_6 - A_4 A_6)}{A_4}$ , $F_i > 0$ where i = 1, 2, 3 and $F_1F_2 > F_3$. Hence,by Routh-Hurwitz criterion the system (\ref{c}) is asymptotically stable in the neighbourhood of $E_2$. \\

\romannum{4})The Jacobian matrix around $E_3$ is given by \\
$$ J(E_3)= \begin{bmatrix}
b_{11} & b_{12} & b_{13} \\
b_{21} & b_{22}& b_{23} \\
b_{31} & b_{32} & b_{33}
\end{bmatrix}
$$ \\
where, $b_{11}= 1-2B_2-D_1$ ,  $b_{12}= -B_2$ , $b_{13}=-B_2,$ 
$b_{21}= 0$, $b_{22}=\frac{A_4B_2}{A_1+D_1}-A_2D_1-A_3$, $a_{23}= 0,$
$b_{31}=A_5D_1$, $b_{32}=A_8D_1B_2 -A_6D_1$, and $b_{33}=0$. \\
The characteristic equation of  J($E_3$) is\\
$$\Lambda^3 + H_1 \Lambda^2 + H_2\Lambda + H_3 = 0,$$
where, $H_1=-(b_{11}+b_{22})$ , $H_2=b_{11} b_{22}-b_{13}b_{31}$, and $H_3=b_{13}b_{31}b_{22}$. \\

Taking $
A_5> A_7 $
and
\small$A_4\le \frac{A_2 A_5^2-2A_2A_5A_7+A_2A_7^2+A_3A_5^2-A_3A_5A_7}{A_5 A_7}$\normalsize 
, we get $H_i > 0$ where i = 1, 2, 3 and $H_1H_2 > H_3$. Hence, by Routh-Hurwitz criteria, the system (\ref{c}) is asymptotically stable in the neighbourhood of $E_3$. \\

\romannum{5})The Jacobian matrix near $E^*$ is given by \\
$$ J(E^*)= \begin{bmatrix}
c_{11} & c_{12} & c_{13} \\
c_{21} & c_{22}& c_{23} \\
c_{31} & c_{32} & c_{33}
\end{bmatrix}
$$ \\
where, \small$c_{11}= 1-2S^{*}-P_1^*-P_2^* ,  c_{12}= -S^{*} , c_{13}=-S^{*}$\normalsize , \small$c_{21}= \frac{A_4P_1^*}{A_1+P_2^*} , c_{22}= \frac{A_4S^{*}}{A_1+P_2^*}-A_2P_2^*-A_3 , 
c_{23}= -\frac{A_4S^{*}P_1^*}{(A_1+P_2^*)^2}-A_2P_1^*,$ \normalsize \\ \small$c_{31}=A_5P_2^*+A_8P_1^*P_2^* , c_{32}=-A_6P_2^*+A_8P_2^*S^{*}$ , and $c_{33}=A_5S^{*}-P_1^*A_6-A_7+A_8P_1^*S^{*}$\normalsize \\

The characteristic equation $J(E^*)$ is 
\begin{equation}\label{f}
    \lambda^3 + G_1 \lambda^2 + G_2\lambda + G_3 = 0
\end{equation} 
where, $G_1=-(c_{11}+c_{22}+c_{33})$ , $G_2=(c_{11} c_{33}+c_{11}c_{22}+c_{22}c_{33}-c_{12}c_{21}-c_{13}c_{31}-c_{23}c_{32})$ , and $G_3=-(c_{11}(c_{22}c_{33}-c_{23}c_{32})+c_{12}(c_{23}c_{31}-c_{21}c_{33})+c_{13}(c_{21}c_{32}-c_{22}c_{31}))$. \\

By  Routh-Hurwitz
criterion, $E^*$ is said to be locally asymptotically stable if and only if the following conditions are met: 
$G_1 >0 $, $G_3>0$, and $G_1G_2-G_3>0$ \\

\end{proof}

\begin{center}
\begin{table}[H]
\begin{tabular}{ | m{3cm} | m{5cm}| m{9cm} | } 
 \hline
Equilibrium State & Stability Type & Stability condition  \\ [1ex] 
 \hline\hline
$E_0$ = (0, 0, 0) & Unstable Saddle point & -  \\ 
 \hline
 $E_1$ = (1,0,0) & locally asymptotically stable &  $A_1A_3 > A_4$ and $A_7>A_5$   \\
 \hline
 $E_2$ = ($B_1$,$C_1$,0) & locally asymptotically stable & $A_1A_3 < A_4$,
$A_5\le\frac{A_4 A_6-A_1 A_3 A_6}{A_1 A_3}$, and
$A_8<\frac{A_1 A_3 A_4 A_5+A_1A_3 A_4 A_6-A_4^2 A_6-A_4^2 A_7}{A_1^2 A_3^2 -A_1 A_3 A_4}$\\
 \hline
 $E_3$= ($B_2$,0,$D_1$) & locally asymptotically stable & $A_5>A_7$ 
and
$A_4\le \frac{A_2 A_5^2-2A_2A_5A_7+A_2A_7^2+A_3A_5^2-A_3A_5A_7}{A_5 A_7}$\\
 \hline
 $E^*$=($S^{*},P_1^*,P_2^*$) & locally asymptotically stable  & $ G_1>0,        G_3>0$, and $G_1G_2-G_3>0$  \\ [1ex] 
\hline
\end{tabular}
\caption{\label{tab2}\textit{Types of stability and stability conditions of all the ecologically feasible equilibrium states}}
\end{table}
\end{center}
\textbf{Remark 3.3.1.} \emph{The existence of local stability around $E_1$ eliminates the feasibilities
of $E_2$ as well as $E_3$.} \\


\subsection{\textit{Global Stability}}This section is dedicated to examining the global stability of both the axial and interior equilibria. Theorems related to this are provided as follows:
\begin{theorem}
    The axial equilibrium point $E_1$ is globally asymptotically stable under the sufficient
conditions $A_4 L_2 \varrho_2 \rho_2 < A_1 \Gamma_1 (A_2 L_2 \varrho_1 + A_7 L_3)$ and
 $L_1 (\varrho_2 + \Gamma_2) + L_3 (A_5 + A_8 \varrho_2) \Gamma_2 \rho_2 < A_3 L_2 \varrho_1$.
\end{theorem}
\begin{proof}
     To study the globally asymptotically stability of the biosystem around $E_1$, the following positive definite Lyapunov
function is considered:
$$V(S,P_1,P_2)=L_1(S-S^{*}-S^{*} ln\frac{S}{S^{*}})+L_2 P_1+L_3 P_2$$
Now taking the time derivative of $V(S,P_1,P_2)$ along the solutions of system (\ref{c}), ˙$V(S,P_1,P_2)$ is
given by,

$$\frac{dV(S,P_1,P_2)}{dt}=L_1\frac{(S-S^{*})}{S}\frac{dS}{dt}+L_2 \frac{dP_1}{dt}+L_3 \frac{dP_2}{dt}$$

Let, $$C_1=L_1\frac{(S-S^{*})}{S}\frac{dS}{dt},C_2=L_2\frac{dP_1}{dt},C_3=L_3\frac{dP_2}{dt}$$\\
Now,$$C_1\le -L_1 (P_1 + P_2) (S-1) $$
Let us consider $\rho_1 < S < \rho_2$,$\varrho_1<P_1<\varrho_2$ and $\Omega_1<P_2<\Omega_2$ and then rearranging the terms\\
 $$C_1+C_2+C_3 \le \Upsilon_1 = \Omega_1 + \Omega_2$$ 
where,
\begin{equation*}
\begin{split}
\Upsilon_1 = & L_1 P_1 - A_3 L_2 P_1 + L_1 P_2 - A_7 L_3 P_2 - A_2 L_2 P_1 P_2 - A_6 L_3 P_1 P_2 - 
 L_1 P_1 S - L_1 P_2 S + A_5 L_3 P_2 S + A_8 L_3 P_1 P_2 S + \frac{A_4 L_2 P_1 S}{
 A_1 + P_2},\\
 \Omega_1 = & L_1 P_1 + L_1 P_2+ A_5 L_3 P_2 S + A_8 L_3 P_1 P_2 S + \frac{A_4 L_2 P_1 S}{
 A_1 + P_2}\\
   \leq &  L_1 P_1 + L_1 P_2+ A_5 L_3 P_2 S + A_8 L_3 P_1 P_2 S + \frac{A_4 L_2 P_1 S}{
 A_1}\\
    \leq & L_1 \varrho_2 + L_1 \Gamma_2+ A_5 L_3 \Gamma_2 \rho_2 + A_8 L_3 \varrho_2 \Gamma_2 \rho_2 + \frac{A_4 L_2 \varrho_2 \rho_2}{
 A_1},\\
 \Omega_2= & - A_3 L_2 P_1- A_7 L_3 P_2 - A_2 L_2 P_1 P_2 - A_6 L_3 P_1 P_2 - 
 L_1 P_1 S - L_1 P_2 S\\
\leq & - A_3 L_2 \varrho_1- A_7 L_3 \Gamma_1 - A_2 L_2 \varrho_1 \Gamma_1 - A_6 L_3 \varrho_1 \Gamma_1 - 
 L_1 \varrho_1 \rho_1 - L_1 \Gamma_1 \rho_1 
 \end{split}
\end{equation*}
So, $C_1+C_2+C_3 \le L_1 \varrho_2 + L_1 \Gamma_2+ A_5 L_3 \Gamma_2 \rho_2 + A_8 L_3 \varrho_2 \Gamma_2 \rho_2 + \frac{A_4 L_2 \varrho_2 \rho_2}{
 A_1}- A_3 L_2 \varrho_1- A_7 L_3 \Gamma_1 - A_2 L_2 \varrho_1 \Gamma_1 - A_6 L_3 \varrho_1 \Gamma_1 - 
 L_1 \varrho_1 \rho_1 - L_1 \Gamma_1 \rho_1 <0;$ provided,
 $A_4 L_2 \varrho_2 \rho_2 < A_1 \Gamma_1 (A_2 L_2 \varrho_1 + A_7 L_3)$ and $L_1 (\varrho_2 + \Gamma_2) + L_3 (A_5 + A_8 \varrho_2) \Gamma_2 \rho_2 < A_3 L_2 \varrho_1$.\\
 
 Thus, $\frac{dV}{dt}<0$ under the conditions:
$A_4 L_2 \varrho_2 \rho_2 < A_1 \Gamma_1 (A_2 L_2 \varrho_1 + A_7 L_3)$ and $L_1 (\varrho_2 + \Gamma_2) + L_3 (A_5 + A_8 \varrho_2) \Gamma_2 \rho_2 < A_3 L_2 \varrho_1.$ 
\end{proof}

\begin{theorem}
    The coexisting equilibrium point $E^*$ is globally asymptotically stable under the sufficient
conditions $P_{2}^{*} \le \Theta_2$ ,
$L_2 > \Upsilon_2$, $A_5 \vartheta_1 (P_{1}^{*} P_{2}^{*} (A_6 - A_8 S^{*}) + A_5 (\vartheta_1 \theta_1 - S^{*} \vartheta_2)) > 0$, $A_4 <\Upsilon_3$, and
 $\vartheta_1 > \Upsilon_4$.
\end{theorem}
\begin{proof}
    To study the globally asymptotically stability of the biosystem around $E^*$, the following positive definite Lyapunov
function is considered:\\
$$W(S,P_1,P_2)=L_1(S-S^{*}-S^{*} ln\frac{S}{S^{*}})+L_2(P_1-P_{1}^*-P_{1}^* ln\frac{P_1}{P_{1}^*})+L_3(P_2-P_{2}^*-P_{2}^* ln\frac{P_2}{p_{2}^*})$$
Now taking the time derivative of $W(S,P_1,P_2)$ along the solutions of system (\ref{c}), we get \\
$$\frac{dW(S,P_1,P_2)}{dt}=L_1\frac{(S-S^{*})}{S}\frac{dS}{dt}+L_2\frac{(P_1-P_{1}^*)}{P_1}\frac{dP_1}{dt}+L_3\frac{(P_2-P_{2}^*)}{P_2}\frac{dP_2}{dt}$$
Let, $D_1=L_1\frac{(S-S^{*})}{S}\frac{dS}{dt}$ , $D_2=L_2\frac{(P_1-P_{1}^*)}{P_1}\frac{dP_1}{dt}$ , $D_3=L_3\frac{(P_2-P_{2}^*)}{P_2}\frac{dP_2}{dt}$\\
Now, $D_1\le -L_1 (P_1 - P_{1}^* + P_2 - P_{2}^*) (S - S^{*}),$
$D_2=L_2 (P_1 -P_{1}^* ) (A_2 (-P_2 +P_{2}^*) + \frac{(
   A_4 (P_{2}^* S + A_1 (S - S^{*}) - P_2 S^{*}))}{((A_1 + P_2) (A_1 +P_{2}^*))}),$
and  $ D_3=-L_3 (P_2 - P_{2}^*) (A_6 (P_1 - P_{1}^* ) - A_8 P_1 S + A_8  P_{1}^* S^{*} + A_5 (-s + S^{*}))$.\\
Taking $L_1=A_5 L_3$,
\begin{equation*}
\begin{split}
D_1+D_2+D_3 \le &  -A_2 L_2 P_1 P_2 - A_6 L_3 P_1 P_2 + A_2 L_2 P_{1}^{*} P_2 + A_6 L_3 P_{1}^{*} P_2 + 
 A_2 L_2 P_1 P_{2}^{*} + A_6 L_3 P_1 P_{2}^{*} - A_2 L_2 P_{1}^{*} P_{2}^{*}\\
 & - A_6 L_3 P_{1}^{*} P_{2}^{*} - 
 A_5 L_3 P_1 S + A_5 L_3 P_{1}^{*} S + A_8 L_3 P_1 P_2 S - A_8 L_3 P_1 P_{2}^{*} S + A_5 L_3 P_1 S^{*}+  + \upsilon_1  + \upsilon_2\\
 & - A_5 L_3 P_{1}^{*} S^{*} - A_8 L_3 P_{1}^{*} P_2 S^{*} + A_8 L_3 P_{1}^{*} P_{2}^{*} S^{*} \\
   \le & -A_2 L_2 P_1 P_2 - A_6 L_3 P_1 P_2 + A_2 L_2 P_{1}^{*} P_2 + A_6 L_3 P_{1}^{*} P_2 + 
 A_2 L_2 P_1 P_{2}^{*} + A_6 L_3 P_1 P_{2}^{*} - A_2 L_2 P_{1}^{*} P_{2}^{*} \\
 & - A_6 L_3 P_{1}^{*} P_{2}^{*} - 
 A_5 L_3 P_1 S + A_5 L_3 P_{1}^{*} S + A_8 L_3 P_1 P_2 S - A_8 L_3 P_1 P_{2}^{*} S  + A_5 L_3 P_1 S^{*} +\upsilon_3 +\upsilon_4\\
 & - A_5 L_3 P_{1}^{*} S^{*} - A_8 L_3 P_{1}^{*} P_2 S^{*} + A_8 L_3 P_{1}^{*} P_{2}^{*} S^{*} 
 \end{split}
\end{equation*}

 Now,Let us consider $\theta_1 < S < \theta_2,  \vartheta_1<P_1<\vartheta_2,$ and $\Theta_1<P_2<\Theta_2$, then rearranging the terms\\
 $$D_1+D_2+D_3 \le \delta_1 + \delta_2$$ 
 where,
\begin{equation*}
\begin{split}
\delta_1 = & A_2 L_2 P_{1}^{*} P_2 + A_6 L_3 P_{1}^{*} P_2 + A_2 L_2 P_1 P_{2}^{*} + A_6 L_3 P_1 P_{2}^{*} + 
 A_5 L_3 P_{1}^{*} S + A_8 L_3 P_1 P_2 S +\upsilon_5  + A_5 L_3 P_1 S^{*} + 
 A_8 L_3 P_{1}^{*} P_{2}^{*} S^{*}\\
 \leq & A_2 L_2 P_{1}^{*} \Theta_2 + A_6 L_3 P_{1}^{*} \Theta_2 + A_2 L_2 \vartheta_2 P_{2}^{*} + A_6 L_3 \vartheta_2 P_{2}^{*} + 
 A_5 L_3 P_{1}^{*} \theta_2 + A_8 L_3 \vartheta_2 \Theta_2 \theta_2 +\upsilon_6 + A_5 L_3 \vartheta_2 S^{*} + 
 A_8 L_3 P_{1}^{*} P_{2}^{*} S^{*}\\
 = & \delta_{1}^{'},\\
 \delta_2= &-A_2 L_2 P_1 P_2 - A_6 L_3 P_1 P_2 - A_2 L_2 P_{1}^{*} P_{2}^{*} - A_6 L_3 P_{1}^{*} P_{2}^{*} - 
 A_5 L_3 P_1 S - A_8 L_3 P_1 P_{2}^{*} S +\upsilon_7 - A_5 L_3 P_{1}^{*} S^{*} - 
 A_8 L_3 P_{1}^{*} P_2 S^{*}\\
 \leq &-A_2 L_2 \vartheta_1 \Theta_1 - A_6 L_3 \vartheta_1 \Theta_1 - A_2 L_2 P_{1}^{*} P_{2}^{*} - A_6 L_3 P_{1}^{*} P_{2}^{*} - 
 A_5 L_3 \vartheta_1 \theta_1 - A_8 L_3 \vartheta_1 P_{2}^{*} \theta_1 +\upsilon_8- A_5 L_3 P_{1}^{*} S^{*} - A_8 L_3 P_{1}^{*} \Theta_1 S^{*}\\
 = & \delta_{2}^{'}
 \end{split}
\end{equation*}
and, $\upsilon_1=\frac{(
 A_1 A_4 L_2 P_1 S)}{(A_1 + P_2) (A_1 + P_{2}^{*})}-\frac{(
 A_1 A_4 L_2 P_{1}^{*} S)}{(A_1 + P_2) (A_1 + P_{2}^{*})} + \frac{(
 A_4 L_2 P_1 P_{2}^{*} S)}{(A_1 + P_2) (A_1 + P_{2}^{*})} - \frac{(
 A_4 L_2 P_{1}^{*} P_{2}^{*} S)}{(A_1 + P_2) (A_1 + P_{2}^{*})}$,
 $\upsilon_2= - \frac{(
 A_1 A_4 L_2 P_1 S^{*})}{(A_1 + P_2) (A_1 + P_{2}^{*})} + \frac{(
 A_1 A_4 L_2 P_{1}^{*} S^{*})}{(A_1 + P_2) (A_1 + P_{2}^{*})} - \frac{(
 A_4 L_2 P_1 P_2 S^{*})}{(A_1 + P_2) (A_1 + P_{2}^{*})} + \frac{(
 A_4 L_2 P_{1}^{*} P_2 S^{*})}{(A_1 + P_2) (A_1 + P_{2}^{*})}$ , $\upsilon_3= \frac{(
 A_1 A_4 L_2 P_1 S)}{(A_1) (A_1 + P_{2}^{*})}-\frac{(
 A_1 A_4 L_2 P_{1}^{*} S)}{(A_1 + P_2) (A_1 + P_{2}^{*})} + \frac{(
 A_4 L_2 P_1 P_{2}^{*} S)}{(A_1) (A_1 + P_{2}^{*})} - \frac{(
 A_4 L_2 P_{1}^{*} P_{2}^{*} S)}{(A_1 + P_2) (A_1 + P_{2}^{*})}$  , $\upsilon_4= - \frac{(
 A_1 A_4 L_2 P_1 S^{*})}{(A_1 + P_2) (A_1 + P_{2}^{*})} + \frac{(
 A_1 A_4 L_2 P_{1}^{*} S^{*})}{(A_1) (A_1 + P_{2}^{*})} - \frac{(
 A_4 L_2 P_1 P_2 S^{*})}{(A_1 + P_2) (A_1 + P_{2}^{*})} + \frac{(
 A_4 L_2 P_{1}^{*} P_2 S^{*})}{(P_2) (A_1 + P_{2}^{*})}$ , $\upsilon_5=\frac{(A_1 A_4 L_2 P_1 S)}{(A_1) (A_1 + P_{2}^{*})} +\frac{(
 A_4 L_2 P_1 P_{2}^{*} S)}{(A_1) (A_1 + P_{2}^{*})}+ \frac{(A_1 A_4 L_2 P_{1}^{*} S^{*})}{(A_1) (A_1 + P_{2}^{*})} + \frac{(
 A_4 L_2 P_{1}^{*} P_2 S^{*})}{(P_2) (A_1 + P_{2}^{*})} $ , $\upsilon_6=\frac{(
 A_1 A_4 L_2 \vartheta_2 \theta_2)}{(A_1) (A_1 + P_{2}^{*})} + \frac{(
 A_4 L_2 \vartheta_2 P_{2}^{*} \theta_2)}{(A_1) (A_1 + P_{2}^{*})} + \frac{(A_1 A_4 L_2 P_{1}^{*} S^{*})}{(A_1) (A_1 + P_{2}^{*})} + \frac{(
 A_4 L_2 P_{1}^{*} \Theta_2 S^{*})}{(\Theta_2) (A_1 + P_{2}^{*})} $ , $\upsilon_7= - \frac{(
 A_1 A_4 L_2 P_{1}^{*} S)}{(A_1 + P_2) (A_1 + P_{2}^{*})} - \frac{(
 A_4 L_2 P_{1}^{*} P_{2}^{*} S)}{(A_1 + P_2) (A_1 + P_{2}^{*})} - \frac{(A_1 A_4 L_2 P_1 S^{*})}{(A_1 + P_2) (A_1 + P_{2}^{*})} - \frac{(
 A_4 L_2 P_1 P_2 S^{*})}{(A_1 + P_2) (A_1 + P_{2}^{*})}$ , and $\upsilon_8= - \frac{(
 A_1 A_4 L_2 P_{1}^{*} \theta_1)}{(A_1 + \Theta_1) (A_1 + P_{2}^{*})} - \frac{(
 A_4 L_2 P_{1}^{*} P_{2}^{*} \theta_1)}{(A_1 + \Theta_1) (A_1 + P_{2}^{*})} - \frac{(A_1 A_4 L_2 \vartheta_1 S^{*})}{(A_1 + \Theta_1) (A_1 + P_{2}^{*})} - \frac{(A_4 L_2 \vartheta_1 \Theta_1 S^{*})}{(A_1 + \Theta_1) (A_1 + P_{2}^{*})}$.\\

Now , $D_1+D_2+D_3 \le \delta_{1}^{'}+\delta_{2}^{'} < 0$ , provided \\

$L_2 > \frac{(L_3 (A_6 P_{2}^{*} \vartheta_2 + A_6 P_{1}^{*} \Theta_2 + A_5 P_{1}^{*} \theta_2 + A_8 \vartheta_2 \Theta_2 \theta_2))}{
 A_2 (\vartheta_1 \Theta_1 - P_{2}^{*} \vartheta_2 + P_{1}^{*} (P_{2}^{*} - \Theta_2))}=\Upsilon_2$,
  $A_5 \vartheta_1 (P_{1}^{*} P_{2}^{*} (A_6 - A_8 S^{*}) + A_5 (\vartheta_1 \theta_1 - S^{*} \vartheta_2)) > 0$,\\
 
 $A_4 < \frac{(A_1 A_6 \vartheta_1 \Theta_1 L_3 (A_1 + P_{2}^{*}))}{
 L_2 (2 A_1 P_{1}^{*} S^{*} + A_1 \vartheta_2 \theta_2 + P_{2}^{*} \vartheta_2 \theta_2)}=\Upsilon_3$,
 $\vartheta_1 > \frac{(P_{2}^{*} \vartheta_2 + P_{1}^{*} (-P_{2}^{*} + \Theta_2))}{\Theta_1}=\Upsilon_4$ , and  
 $P_{2}^{*} \le \Theta_2$\\

Thus, $\frac{dW}{dt} < 0$ under the conditions: $L_2 > \Upsilon_2$, 
 $A_5 \vartheta_1 (P_{1}^{*} P_{2}^{*} (A_6 - A_8 S^{*}) + A_5 (\vartheta_1 \theta_1 - S^{*} \vartheta_2)) > 0$,
  $A_4 <\Upsilon_3$, 
 $\vartheta_1 > \Upsilon_4$, and 
 $ P_{2}^{*} \le \Theta_2$.

\end{proof}

\section{\textit{Qualitative changes in the ecological scenario}}
Stability or direction of stability may change in an ecosystem once the influencing parameters are varied. Bifurcations within an ecological system hold significance due to their potential to instigate the emergence of novel behaviours within said system.We theoretically discuss some of the bifurcations of codimension one in this section.
\subsection{Transcritical bifurcation} The phenomenon of transcritical bifurcation holds considerable significance within the realm of dynamical systems and nonlinear dynamics. It is a phenomenon that arises when a parameter traverses a critical value, leading to a change in the stability of an equilibrium point. Theorems relevant to this are provided. 
\begin{theorem}
    (\romannum{1}) \textit{The system (\ref{c}) experiences a transcritical bifurcation  as the equilibrium point $E_1$ crosses the critical value $A_5$=$A_7$ provided $A_5 \ne 0$}.\\
    
(\romannum{2}) \textit{The system (\ref{c}) undergoes a transcritical bifurcation  around the  equilibrium point $E_2$ along the parametric surface $A_1 A_3 A_4(A_5 + A_6 + A_8)-A_4^2(A_6 + A_7)-A_1^2 A_3^2 A_8 =0$}. \\

(\romannum{3}) \textit{The system (\ref{c}) exhibits a transcritical bifurcation along the parametric surface $-A_2 A_5^2 - A_1 A_2 A_5^2 - A_3 A_5^2 - A_1 A_3 A_5^2 + 2 A_2 A_5 A_7 + 
 A_1 A_2 A_5 A_7 + A_3 A_5 A_7 + A_4 A_5 A_7 - A_2 A_7^2=0$  around the  equilibrium point $E_3$.}
\end{theorem}
\begin{proof}
     (\romannum{1}) At $A_5$=$A_7$=$A_5^{tb}$, The system (\ref{c}) experiences a transcritical bifurcation around the equilibrium point $E_1$ as $J(E_1)$ has one eigenvalue zero. The Jacobian matrix at the equilibrium point $E_1$ at $A_5=A_5^{tb}$ is

$$J(E_1)=  \begin{bmatrix}
-1 & -1 & -1 \\
0 & \frac{A_4 -A_1A_3}{A_1} & 0 \\
0 & 0 & 0
\end{bmatrix}
$$ 

Now, $P=(p_1,p_2,p_3)^t=(-1,0,1)^t$ , and $Q=(q_1,q_2,q_3)^t=(0,0,1)^t$ are two eigenvectors corresponding to the zero eigenvalue of the matrices $J(E_1)$and $J(E_1)^T$ respectively. After some calculation,we get \\

$K_{A_5}(E_1;A_5^{tb})=\begin{bmatrix}
0 \\
0 \\
0 
\end{bmatrix}
, D(K_{A_5}(E_1;A_5^{tb}))P=\begin{bmatrix}
0 & 0 & 0\\
0 & 0 &0\\
0 & 0 & 1
\end{bmatrix}
\begin{bmatrix}
-1 \\
0\\
1
\end{bmatrix}=\begin{bmatrix}
0 \\
0 \\
1 
\end{bmatrix}, $ and $ D^2(K_{A_5}(E_1;A_5^{tb}))(P,P)=\begin{bmatrix}
0\\
0\\
-A_5 
\end{bmatrix}
$ .\\
Therefore,
$Q^T(K_{A_5}(E_1;A_5^{tb}))=0, Q^T(D(K_{A_5}(E_1;A_5^{tb}))P)=1 \ne 0$ , and
$Q^T(D^2(K_{A_5}(E_1;A_5^{tb}))(P,P))=-A_5 \ne 0$.\\

Hence, at  $A_5=A_5^{tb}$, a transcritical bifurcation occurs around $E_1$, although  other parameters can also be taken as bifurcating parameters.\\
The proofs for (\romannum{2}) and (\romannum{3}) can be derived in a similar manner to that of (\romannum{1}).
\end{proof}








\subsection{\textit{Hopf bifurcation}} The Hopf bifurcation is a fundamental concept within the wider domain of bifurcation theory. It offers significant insights into the shift from stable behaviour to oscillatory behaviour. It serves as a crucial foundation for investigating more intricate bifurcation scenarios. Below are the theorems pertaining to Hopf bifurcation.
\begin{theorem} \label{t8}
    The system  (\ref{c}) does not experience Hopf bifurcation for any parameter at the
equilibrium points $E_2$ and $E_3$.
\end{theorem}
\begin{proof}
     According to Liu \cite{liu}, Hopf bifurcation around $E_2$ for a parameter value say $\mu$=$\mu_0$ exists iff conditions :\\

(\romannum{1})$F_i(\mu_0)>0$ for i=1,2,3 and     $(V)(\mu_0)=0$ ,where $V=F_1F_2-F_3$.\\

(\romannum{2})$\frac{dV}{d\mu}(\mu_0)\neq 0$ both hold.\\

 Solving V=0 , we get- \\
 
$A_7=\frac{-2A_1^2A_3^2A_8+A_1A_3A_4(2A_5+2A_6+2A_8-1)-Z_1}{2A_4^2}$ and $A_7=\frac{-2A_1^2A_3^2A_8+A_1A_3A_4(2A_5+2A_6+2A_8-1)+Z_1}{2A_4^2}$\\

where, $Z_1=A_4(\sqrt{A_3(A_1^2A_3+4A_1A_3A_4-4A_4^2)}+2A_4A_6)$\\

putting first value of $A_7$ in $F_1$ ,it is found that\\

$F_1>0	\Leftrightarrow \frac{4 \text{$A_4$}^2}{\text{$A_1$}^2+4 \text{$A_1$} \text{$A_4$}}\leq \text{$A_3$}<\frac{\text{$A_4$}}{\text{$A_1$}}$ and
$F_2< 0 \Leftrightarrow A_3 \ge \frac{4A_4^2}{A_1^2+4A_1A_4}$
Hence, $F_1$, $F_2$ cannot be positive simultaneously and
hence the result.\\


Similarly, we can show that conditions (\Romannum{1}) or (\Romannum{2}) or both do not hold for the latter value of $A_7$, which implies the fact that the system (\ref{c}) does not experience Hopf bifurcation  around the equilibrium point $E_2$.\\

In a similar way, it can be shown that  the system (\ref{c}) does not experience Hopf bifurcation  around the equilibrium point $E_3$ as well.

\end{proof}

\begin{theorem} \label{t9} The system (\ref{c}) undergoes Hopf bifurcation  around the equilibrium point $E^{*}$ when $A_4$ passes through the critical value $A_4^{H*}$ where $A_4^{H*}$ satisfies G($A_4^{H*}$)=$G_1(A_4^{H*})$ $G_2(A_4^{H*})$-$G_3(A_4^{H*})$=0,$G_i(A_4^{H*})>0$ for i=1,2,3 and $G_1(A_4^{H*})$ $G_2^{'}$($A_4^{H*}$) +$G_2(A_4^{H*})$ $G_1^{'}$($A_4^{H*}$) -$G_3^{'}$($A_4^{H*}$)   $ \neq 0$. Where $G_1, G_2$ and $G_3$ have their usual meanings.
\end{theorem}

\begin{proof}
    By the condition G($A_4$)=0 at $A_4$=$A_4^{H*}$, The Characteristic equation (\ref{e}) can be written as\\
\begin{equation}\label{k}
    (\lambda^2+G_2)(\lambda+G_1)=0 
\end{equation}

Thus we have , $\lambda$= -$G_1$ , $\pm$ $\sqrt{G_2}$ i; provided $G_1$,$G_2>0$. For $A_4$ $\in$ ($A_4^{H*}$-$\epsilon$ ,$A_4^{H*}$+$\epsilon$), the general form of the roots are 
$\lambda_{1}=\zeta_1(A_4) + \zeta_2(A_4)$,
$\lambda_{2}=\zeta_1(A_4) - \zeta_2(A_4)$ and
$\lambda_3=-G_1(A_4)$
Now, substituting $\lambda_{1}=\zeta_1(A_4^{H*}) + \zeta_2(A_4^{H*})$ in equation $(\ref{f})$ and then differentiating and separating the real and imaginary part ,we get \\
\begin{equation} \label{g}
  \begin{split}
& P(A_4) {\zeta_1}^{'}(A_4)-Q(A_4)\zeta_2^{'}(A_4)+R(A_4)=0  \\
& Q(A_4)\zeta_1^{'}(A_4)+P(A_4)\zeta_2^{'}(A_4)+S(A_4)=0
 \end{split}
\end{equation}

where, $P(A_4)=3\zeta_1^{2}-3\zeta_2^{2}+ G_2+2G_1\zeta_1$ ,
 $Q(A_4)=6\zeta_1\zeta_2+2G_1\zeta_2$ ,
and $R(A_4)=G_1^{'}\zeta_1^{2}-G_1^{'}\zeta_2^{2}+G_2^{'}\zeta_1+G_3^{'}$ ,
$S(A_4)=G_2^{'}\zeta_2+2\zeta_1\zeta_2$.
From System (\ref{g}) , we can show that-
$${\zeta_1}^{'}(A_4)=- \frac{P(A_4) R(A_4)+Q(A_4)S(A_4)}{P^2(A_4)+Q^2(A_4)}$$
Using $\zeta_1^{'}(A_4^{H*})=0$,$\zeta_2^{'}(A_4^{H*})=i\sqrt{G_2(A_4^{H*})}$, we get-
$P(A_4^{H*})= - 2G_2$ , 
$Q(A_4^{H*})=  2G_1\sqrt{G_2}$ , 
and $R(A_4^{H*})=G_3^{'} - G_1^{'}G_2$ , 
$S(A_4^{H*})=\sqrt{G_2}G_2^{'}$
 Now, for Transversality condition,we need to show that-
$$\frac{d}{dA_4}(Re(\lambda_i(A_4)))_{A_4=A_4^{H*}} \neq 0; i=1,2,3$$ 

Hence, the following equation is obtained\\
\begin{equation}
\begin{split}
\frac{d}{dA_4}(Re(\lambda_i(A_4)))_{A_4=A_4^{H*}} &=- \frac{P(A_4) R(A_4)+Q(A_4)S(A_4)}{P^2(A_4)+Q^2(A_4)}\\
& = \frac{-2(G_3^{'}-G_1^{'}G_2)+2G_1 G_2 G_2^{'}}{4(G_2^2+G_1^2G_2)}\\
& = \frac{G_1^{'}G_2-G_3^{'}+G_1G_2{'}}{2(G_2+G_1^2)}\\
\end{split}
\end{equation}
Now, if $G_1^{'}G_2-G_3^{'}+G_1G_2{'}\neq 0$, then the transversality condition is satisfied and hence the abovementioned system (\ref{c}) experiences Hopf-bifurcation around the equilibrium point $E^*$ when $A_4$ passes through the critical value $A_4=A_4^{H*}$. The presence of Hopf bifurcation for other parameters can be demonstrated in a similar manner. Discussion pertaining to the direction and stability of Hopf bifurcations for different parameters are done numerically in section 5.2.1.
 \\
\end{proof}

\subsubsection{\textit{Direction  and stability of Hopf bifurcations around the interior equilibrium point $E^*$}}
\begin{theorem}\label{t10} \cite{has} The sign of $ \mu_2\alpha^{'}(0)$ clarifies the direction of Hopf
bifurcation. The biosystem (\ref{c}) encounters supercritical Hopf bifurcation if $\mu_2\alpha^{'}(0)>0$ and subcritical Hopf bifurcation when
$\mu_2\alpha^{'}(0)<0$. $\beta_2$ indicates the stability of the bifurcating periodic solution; it is stable for negative $\beta_2$ and unstable for
a positive $\beta_2$. Here, $\alpha^{'}(0)$, $\mu_2$, and $\beta_2$ have their usual meanings.
\end{theorem}


\section{\textit{Numerical Simulations}}
In the previous sections,we have established three important equilibrium points $E_2$ (second predator free) and $E_3$ (first predator free) as well as coexistence equilibrium point $E^*$. In this section, we perform numerical simulations of system (\ref{c}) to illustrate the results obtained from our theoretical analysis using MATHEMATICA , MATLAB , as well as MATCONT \cite{Dhoo}. Some ecologically possible parametric values have been thought of, as shown in the table below:\\

\begin{table}[H]
\centering
\begin{tabular}{|c c|} 
  \hline
 Parameter & value  \\ [0.5ex] 
 \hline\hline
 $A_1$ & 0.02 \\ 
 \hline
 $A_2$ & 0.05  \\
 \hline
 $A_3$ & 2 \\
 \hline
 $A_4$ & 0.4 \\
 \hline
 $A_5$ & 4.536 \\
 \hline
 $A_6$ & 0.052 \\
 \hline
 $A_7$ & 4.546 \\
 \hline
 $A_8$ & 114.98 \\ [1ex] 
 \hline
\end{tabular}
\caption{\emph{Parameter values for the system (\ref{c})}}
\label{table:1}
\end{table}

\begin{figure}[H]
     \centering
     \begin{subfigure}{0.45\textwidth}
         \centering
         \includegraphics[width=\textwidth]{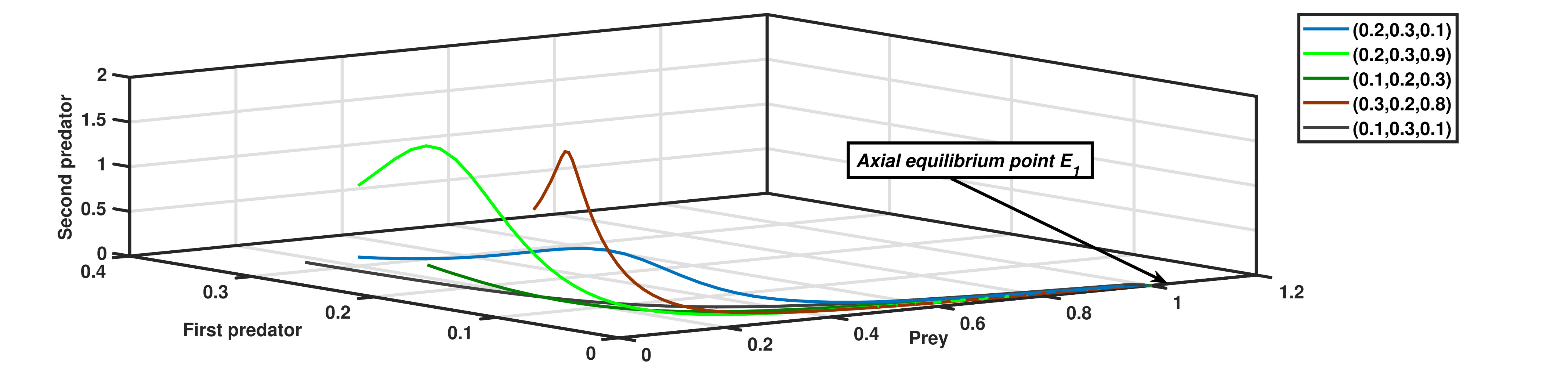}
         \caption{\emph{Axial equilibrium point $E_1$=(1,0,0) is locally asymptotically stable}}
         \label{1a}
     \end{subfigure}
     \hfill
     \begin{subfigure}{0.45\textwidth}
         \centering
         \includegraphics[width=\textwidth]{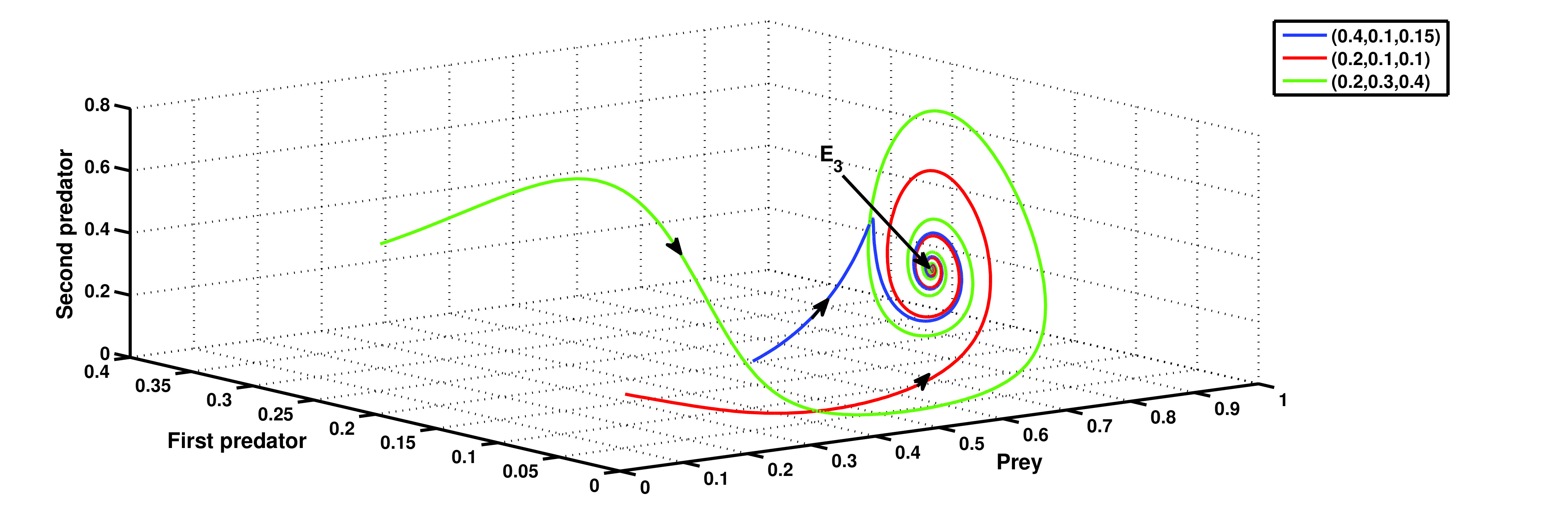}
         \caption{\emph{First predator free equilibrium point $E_3$ is locally asymptotically stable }}
         \label{1b}
     \end{subfigure}
      \hfill
     \begin{subfigure}{0.45\textwidth}
         \centering
         \includegraphics[width=\textwidth]{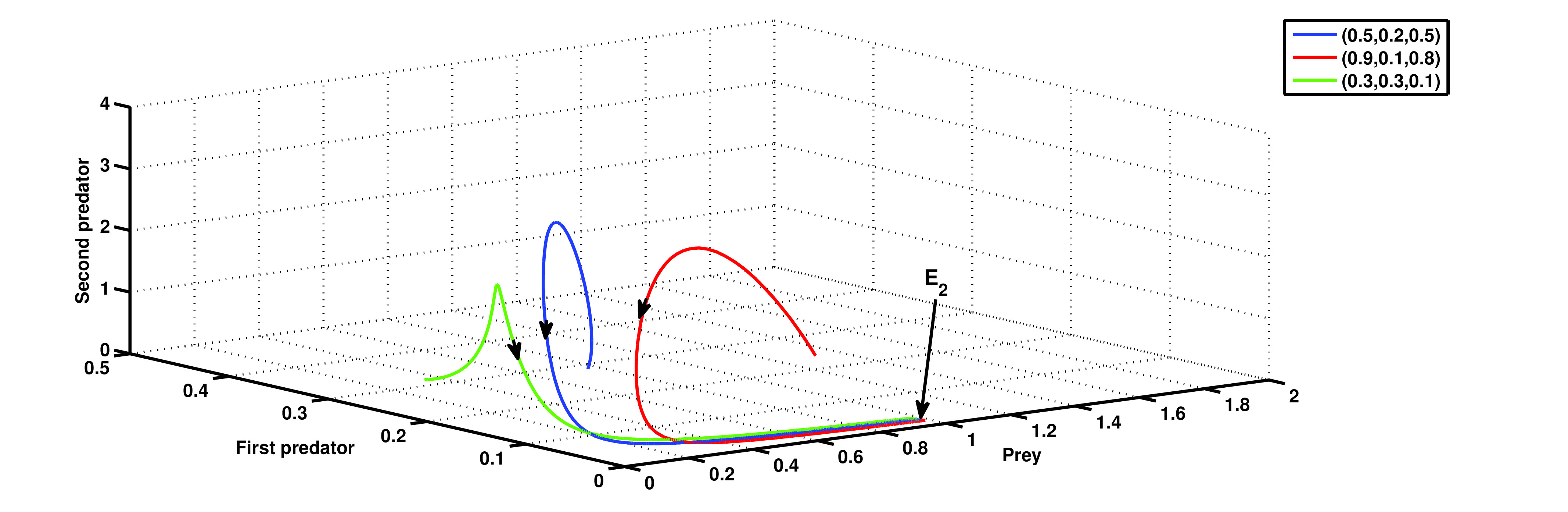}
         \caption{\emph{Second predator free equilibrium point $E_2$ is locally asymptotically stable}}
         \label{1c}
     \end{subfigure}
        \caption{\emph{Local Stability of axial and boundary equilibrium points}}
        \label{lsaxbd}
\end{figure}

\begin{figure}[H]
     \centering
     \begin{subfigure}[F]{0.45\textwidth}
       \centering
         \includegraphics[width=90mm]{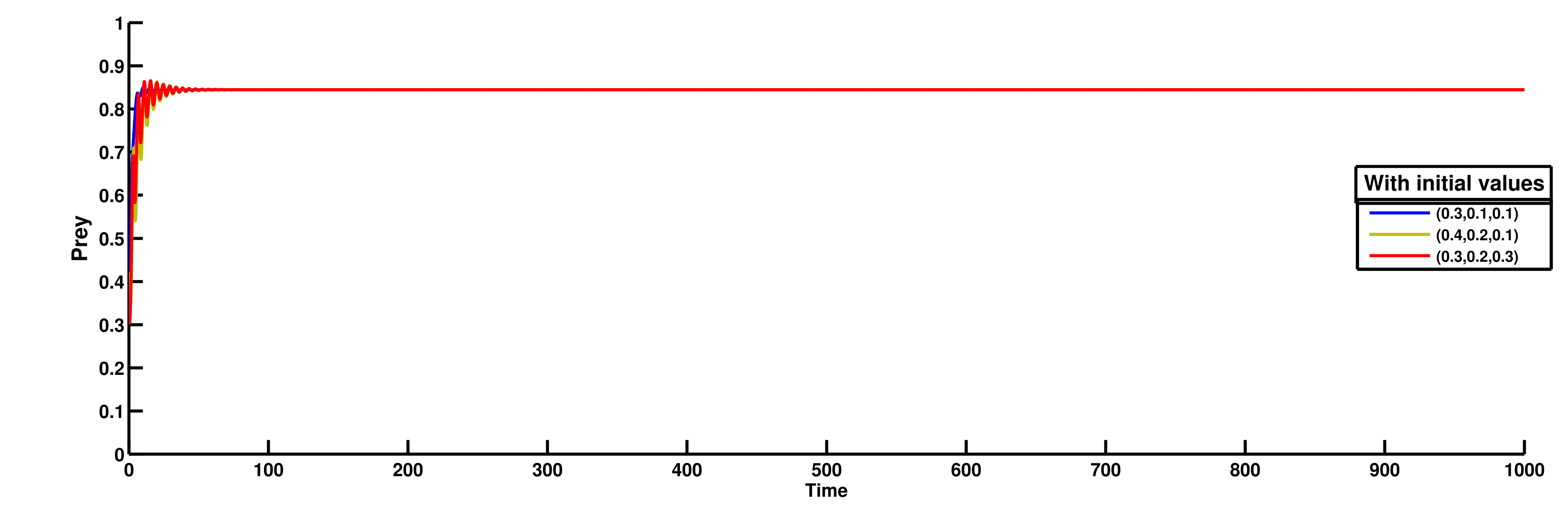}
         \caption{\emph{ Time series of prey  }}
         \end{subfigure}
          \hfill
     \begin{subfigure}[F]{0.45\textwidth}
        \centering
         \includegraphics[width=90mm]{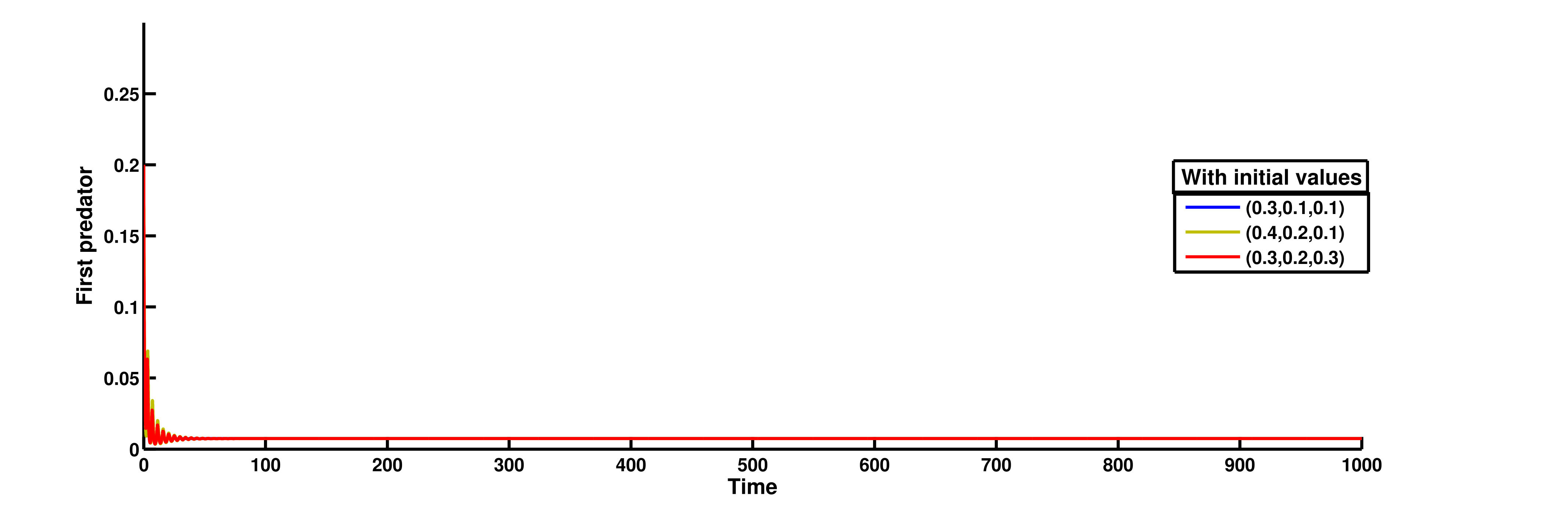}
         \caption{\emph{Time series of first predator}}
         \end{subfigure}
      \hfill
     \begin{subfigure}[F]{0.45\textwidth}
        \includegraphics[width=90mm]{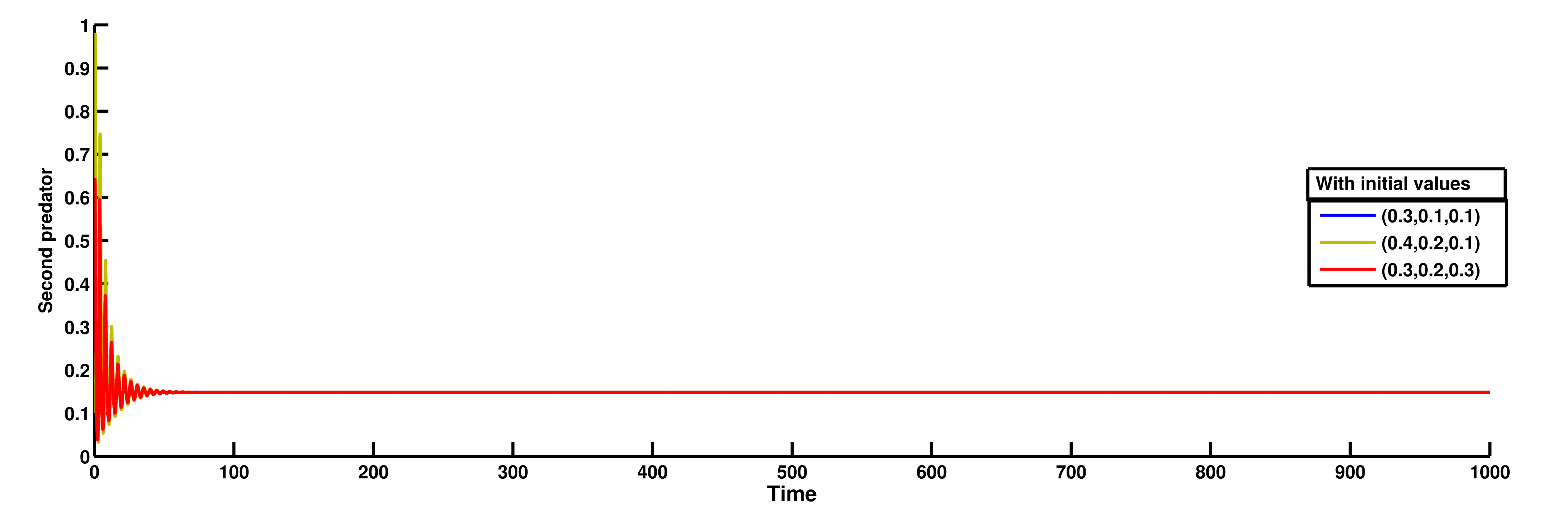}
         \caption{\emph{Time series of second predator} }
    \end{subfigure}
     
     \caption{\emph{Local Stability of coexisting equilibrium $E^*(x^*,y^*, z^*$)} 
}
        \label{lsestar}
\end{figure}

    \subsection{\textit{ Stability of the equilibrium points}}
    
 In this section, at first, we compute and numerically verify the parametric conditions for the existence and stability of the equilibrium points: 
    \par Taking parameter values as $A_3=0.875$, $A_4=0.0125$, $A_5=4.536$ and the rest values from table (\ref{table:1}), the stabilty conditions for local stabilty of the axial equilibrium point $E_1$ are satisfied as $A_1 A_3=0.0175>A_4$ , and $A_7=4.546>A_5$ which can be seen from the figure (\ref{1a}). Various coextant initial populations lead to a predator-free equilibrium point through different trajectories, as shown in figure (\ref{1a}) interpreting the fact that when two predator species go extinct, then prey grows without any disturbances and eventually settles at the carrying capacity of the system, which is found to be true in nature.

To verify the stability conditions of the first predator free equilibrium point $E_3$, we consider $A_3=2.723732$, $A_4=0.516658$ , $A_5=9.3241$ and the rest parameter values from table (\ref{table:1}), all the stability conditions $A_5=9.3241>A_7$ and $\frac{A_2 A_5^2-2A_2A_5A_7+A_2A_7^2+A_3A_5^2-A_3A_5A_7}{A_5 A_7}=2.88978>A_4$ are satisfied for these parameter values. It can be seen in figure (\ref{1b}) confirming the local stabilty of the first predator free equilibrium point $E_3$.

 We consider  $A_3=0.51558$, $A_4=0.010528$, $A_5=1.135458$, and the rest parameters from table (\ref{table:1}) to validate the stability conditions of the second predator free equilibrium point $E_2$. These parameter values satisfy all the stability conditions as  $A_1 A_3=0.0103116<A_4$ , $\frac{A_4 A_6-A_1 A_3 A_6}{A_1 A_3}=0.00109531 < A_5$ , $\frac{A_1 A_3 A_4 A_5+A_1A_3 A_4 A_6-A_4^2 A_6-A_4^2 A_7}{A_1^2 A_3^2 -A_1 A_3 A_4}=170.022 > A_8$, and $  \frac{(A_1 A_3 A_5 + A_1 A_3 A_6 - A_4 A_6)}{A_4} = 1.11096 < A_7 $. Figure (\ref{1c}) is the numerical simulation verifying the local stability criteria of the second predator free equilibrium point $E_2$.
     
In an ecosystem, the most desirable outcome of a biosystem is the stablity of its interior equilibrium. For the same reason, we consider parameter values as $A_3=2$, $A_4=0.4$, 
 $A_5=4.536$ and the rest parameter values from table (\ref{table:1}) which satisfies the local stabilty conditions of the coexistant equilibrium point $E^*$ given by Routh-Hurwitz criterion as $G_1=0.844371>0$, $G_2=1.96046>0$, $G_3=1.22722>0$, and $G_1 G_2-G_3= 0.428135>0$. Numerical simulation satisfying the abovementioned parameter values is shown in figure(2) and it confirms the existence as well as local stability of coexisting equilibrium point $E^*$ and hence, verifies the theoretical results.\\
Figure(\ref{global}) indicates that the axial equilibrium point as well as the cohabitation equilibrium point are globally asymptotically stable. As shown in figure(\ref{3b}), starting from five different initial values, the solution trajectories of system (\ref{c}) ultimately converge to the coexistence equilibrium point $E^*$. All parameter values from the table (\ref{table:1}) are considered here. It confirms the global stability of the cohabitation equilibrium point $E^*$. Morever, {Figure(\ref{3a}) validates the global stability of the axial equilibrium point $E_1$. In this context, except for the values $A_3=0.875$, $A_4=0.0125$, $A_5=4.536$ , all other parameter values shown in the table  (\ref{table:1}) are taken into consideration here.

\begin{figure}[H]
     \centering
     \begin{subfigure}[F]{0.45\textwidth}
       \centering
         \includegraphics[width=90mm]{Figure1a.jpg}
         \caption{\emph{Axial Equilibrium point $E_1$=(1,0,0) is globally asymptotically stable}}
          \label{3a}
         \end{subfigure}
          \hfill
     \begin{subfigure}[F]{0.45\textwidth}
        \centering
         \includegraphics[width=90mm]{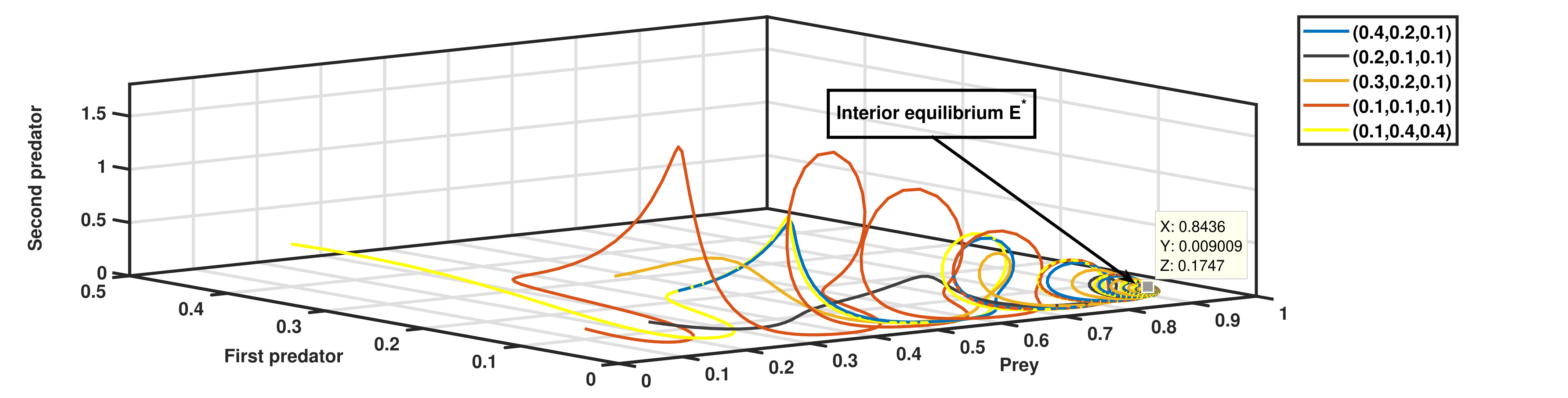}
         \caption{\emph{Interior Equilibrium point $E^*$ is globally asymptotically stable}} \label{3b}
         \end{subfigure}
     \caption{\emph{Phase portraits showing global stability of different equilibrium points} 
}
        \label{global}
\end{figure}

\subsection{\textit{ Parameters affecting stability of equilibrium points}} In this section, the impacts of different parameters on the stability of equilibrium points are analysed.
 \subsubsection{\emph{Role of growth and death factors of both the predators}}
$A_3$ is the parameter related to the death rate of the first predator. The mortality rate of a species is a critical determinant of its ability to persist. The figure (\ref{A3}) is drawn varying $A_3$ up to the value $0.5$ and the rest of the parameter values from table  (\ref{table:1}). It shows that the dynamical system experiences a saddle-node bifurcation at the coexisting equilibrium point when the parameter $A_3 = A_3^S=0.04267$. When $A_3<A_3^S$ is employed, none of the coexisting equilibrium points are  existent. However, when $A_3>A_3^S$, two simultaneous equilibrium points come to exist, one is stable and another is unstable i.e., the coexisting equilibrium point with a relatively lower density of the first predator than the other equilibrium point, becomes stable. It is noteworthy that reducing $A_3$ to $A_3=0.4$ with rest of the parameter values from the table (\ref{table:1}) induces extinction of second predator as $A_1 A_3=0.008<A_4$ , $\frac{A_4 A_6-A_1 A_3 A_6}{A_1 A_3}=2.548 < A_5$,  $\frac{A_1 A_3 A_4 A_5+A_1A_3 A_4 A_6-A_4^2 A_6-A_4^2 A_7}{A_1^2 A_3^2 -A_1 A_3 A_4}=229.91 > A_8$, and $\frac{(A_1 A_3 A_5 + A_1 A_3 A_6 - A_4 A_6)}{A_4} = 0.03976 < A_7 $. So, bistability of both the coexisting equilibrium as well as the second predator free equilibrium comes to exist. Even considering $A_3=1$ and rest parameter values from table  (\ref{table:1}), shows that all the three species survive as $G_1=0.720216>0$ , $G_2=2.43926>0$ , $G_3=1.15537>0$, and $G_1 G_2-G_3= 0.601422>0$. But, taking $A_3=21$ alongwith the same other parameter values as before then not only first predator vanishes but also second predator faces extinction as $A_1 A_3=0.42>A_4$ and $A_7 = 4.546 > A_5$.\\
 \begin{figure}[H]
         \centering
         \includegraphics[width=\textwidth]{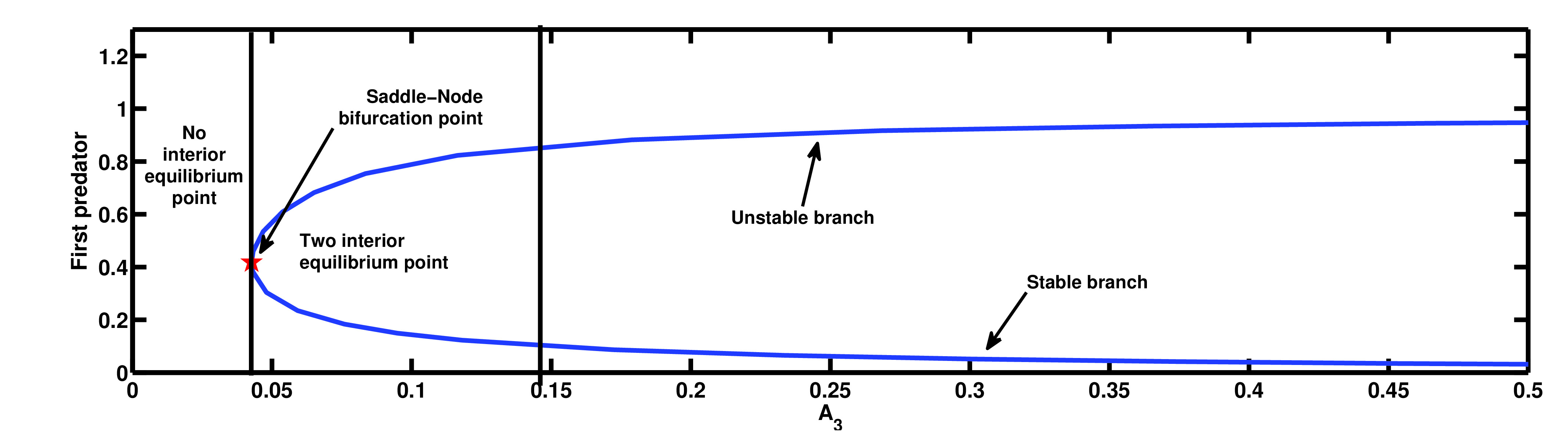}
         \caption{\emph{The emergence of saddle-node bifurcation for parameter $A_3$} }
         \label{A3}
     \end{figure}
     
   $A_4$ is the parameter related to the conversion efficiency or growth rate of the first predator. Figure (\ref{A4}) depicts the progression of the coexisting equilibrium as $A_4$ increases while retaining the other parameter values from table  (\ref{table:1}).
   \begin{figure}[H]
         \centering
         \includegraphics[width=\textwidth]{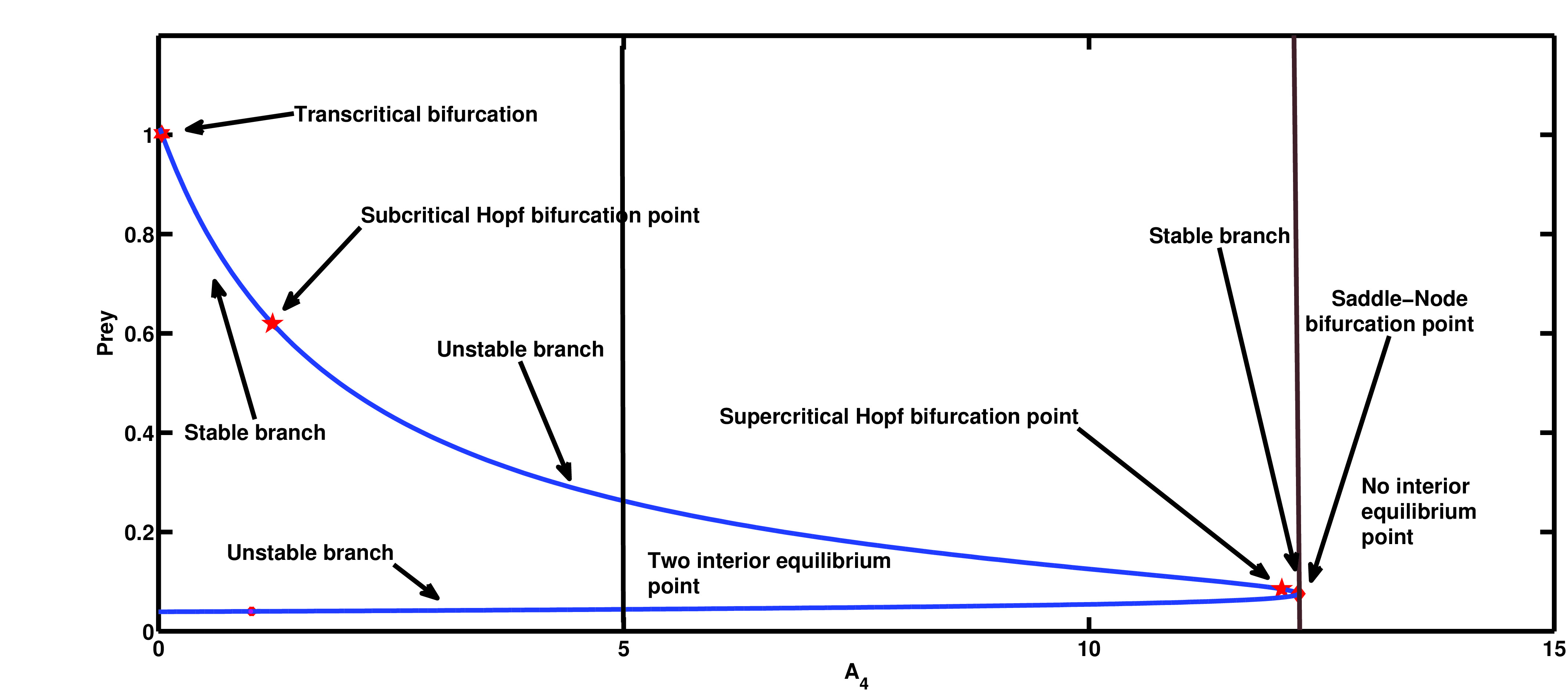}
         \caption{\emph{ The appearance of subcritical Hopf bifurcation, supercritical Hopf bifurcation, saddle-node bifurcation, and transcritical bifurcation for parameter $A_4$}}
         \label{A4}
     \end{figure}
   At $A_4=0$ alongwith other parameter values from table  (\ref{table:1}), at first, the first predator vanishes which in turn results in the extinction of the second predator as $A_1 A_3>A_4$ and $A_7 = 4.546 > A_5$. One of the coexistent equilibrium points having comparatively high prey density becomes an attractive fixed point due to transcritical bifurcation at $A_4=A_4^T=0.04$. Considering $A_4=0.4$, coexistence of all the three species is possible as $G_1=0.844371>0$, $G_2=1.96046>0$, $G_3=1.22722>0$, and $G_1 G_2-G_3= 0.428135>0$.  On the further increase of the same parameter value, the coextant equilibrium point becomes unstable and the system shows a subcritical Hopf bifurcation at the point $A_4=A_4^H=1.2252$  as $G_1(A_4^{H})=0.619461>0$,  $G_2(A_4^{H})=5.01653>0$, $G_3(A_4^{H})=3.10759>0$, $G_1(A_4^{H})$ $G_2(A_4^{H})$-$G_3(A_4^{H})$=0, and $G_1(A_4^{H}) G_2'(A_4^{H}) +G_2(A_4^{H})G_1' (A_4^{H}) -G_3'(A_4^{H}) =129.828 \neq 0$.i.e., it satisfies the NASC (as specified in Theorem (\ref{t9})) for the existence of Hopf bifurcation. 
Procedures mentioned by Hazzard et.al \cite{has} are employed to determine the nature and direction of bifurcating periodic solutions for the mentioned parameter values, and we get
$C_1(0)=0.806671 + 14.3025i$, $\mu_2=-1.23637<0$, $\beta_2=1.61334>0$, and $\alpha^{'}(0)=0.652451>0$.
Negative value of $\mu_2$ alongwith positive values of $\beta_2$, and $\alpha^{'}(0)$ indicate subcritical nature of the Hopf bifurcation (as mentioned in theorem (\ref{t10})). Also the value of first Lyapunov coefficient = $2.177019e^{-1}$. Biologically it means that initially $A_4^{H}-\epsilon<A_4<A_4^{H}$, $\epsilon$ is a small quantity, coextant population follows a spiral trajectory, to reach the coexisting equillibrium which can be seen from figure (\ref{A4}).
The initial population gets to the stable coexisting equilibrium point with the aid of an unstable limit cycle encircling it. The drawback of this unstable limit cycle is that any population that coexists having the population size outside of it will never be able to reach the equilibrium point when the original coexisting population is numerically simulated while maintaining the specified parameter values, figures (\ref{before a4}) and (\ref{after A4}) depict the precise situation.

  \begin{figure}[H]
     \centering
     \begin{subfigure}[F]{0.45\textwidth}
         \centering
         \includegraphics[width=1.2\textwidth]{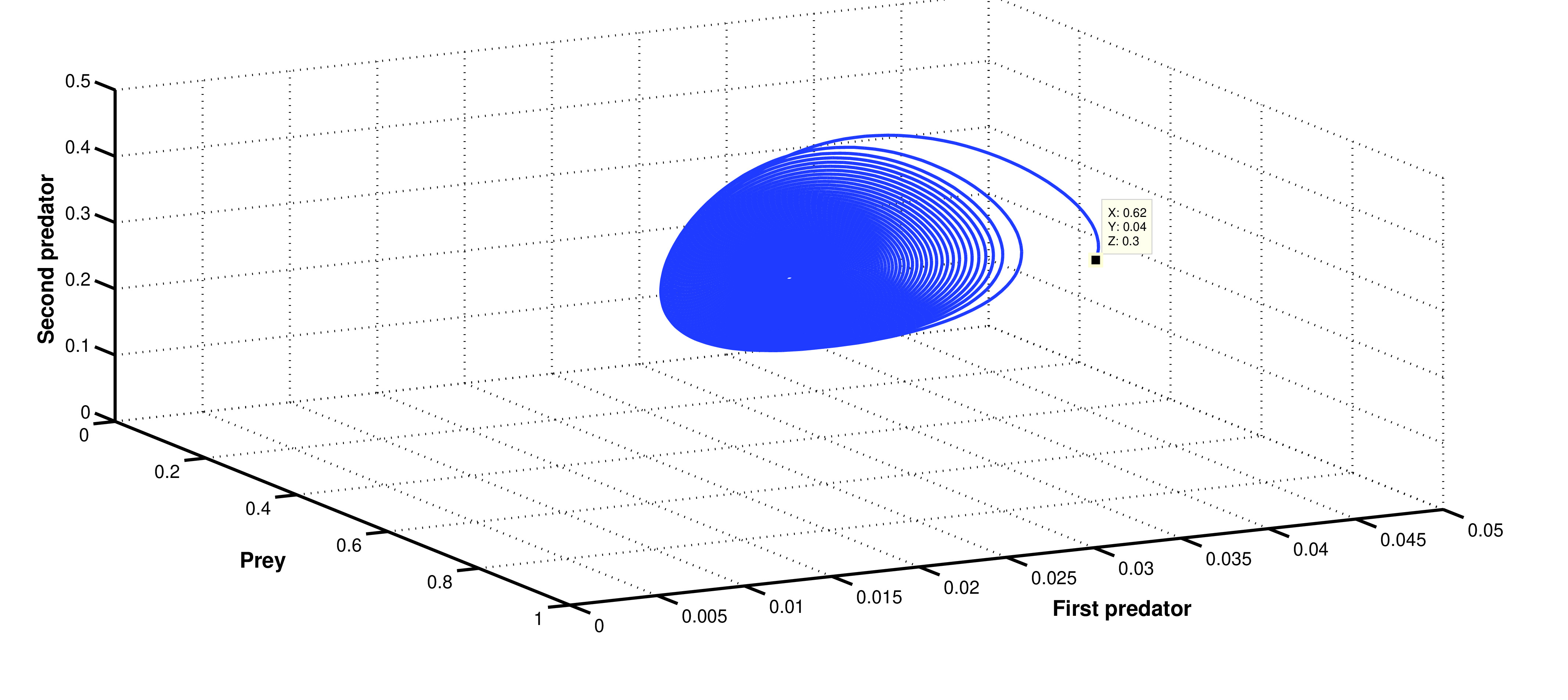}
         \caption{\emph{Phase portrait} }
         \label{}
     \end{subfigure}
     \hfill
     \begin{subfigure}[F]{0.45\textwidth}
         \centering
         \includegraphics[width=\textwidth]{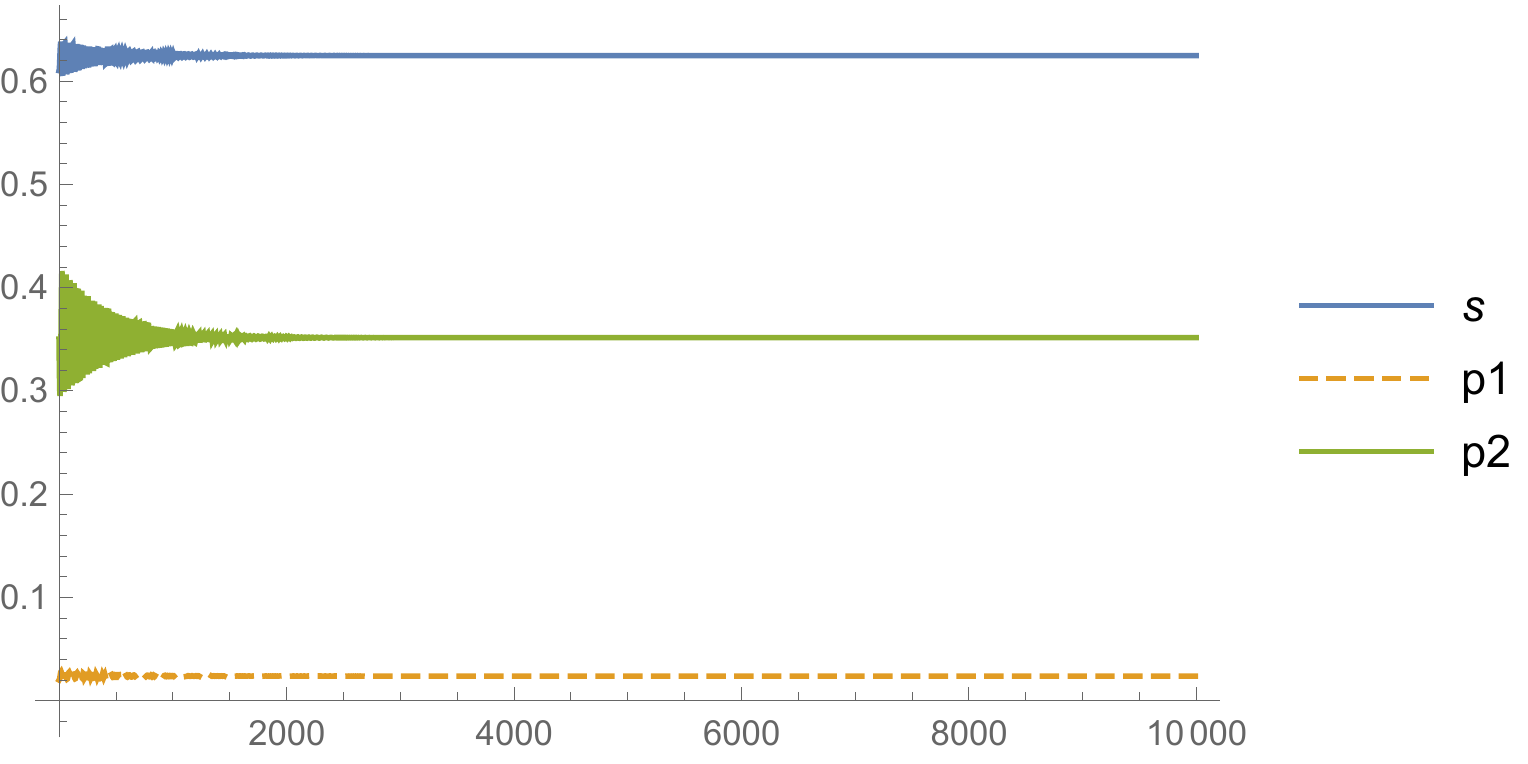}
         \caption{\emph{Time series of all three populations
}}
     \end{subfigure}
        \caption{\emph{The illustration portrays the condition at $A_4<A_4^{H}$, denoting the phase prior to subcritical Hopf bifurcation at $A_4=A_4^{H}=1.2252$.}
}
        \label{before a4}
\end{figure}
\vspace*{1cm}

\begin{figure}[H]
     \centering
     \begin{subfigure}[F]{0.45\textwidth}
         \centering
         \includegraphics[width=1.2\textwidth]{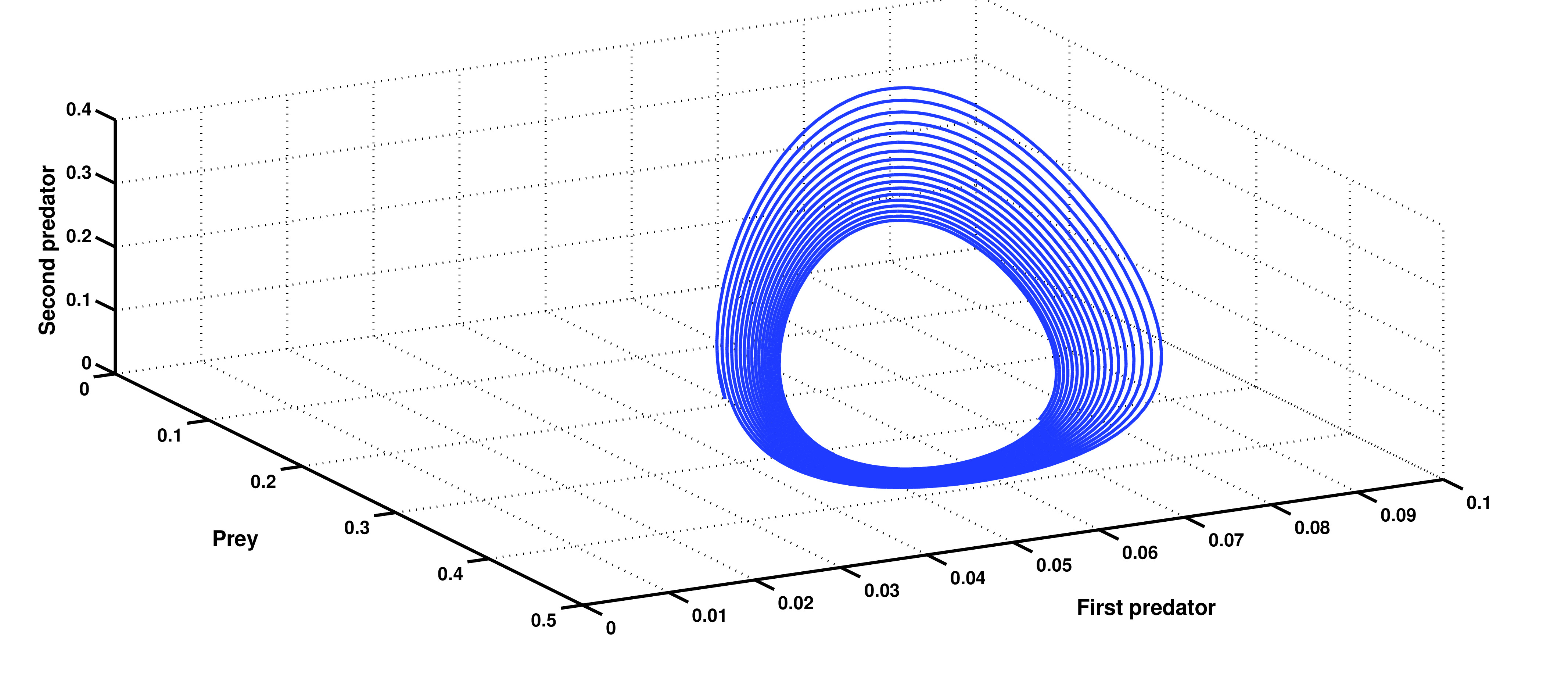}
         \caption{\emph{Phase portrait} }
         \label{ph}
     \end{subfigure}
     \hfill
     \begin{subfigure}[F]{0.45\textwidth}
         \centering
         \includegraphics[width=\textwidth]{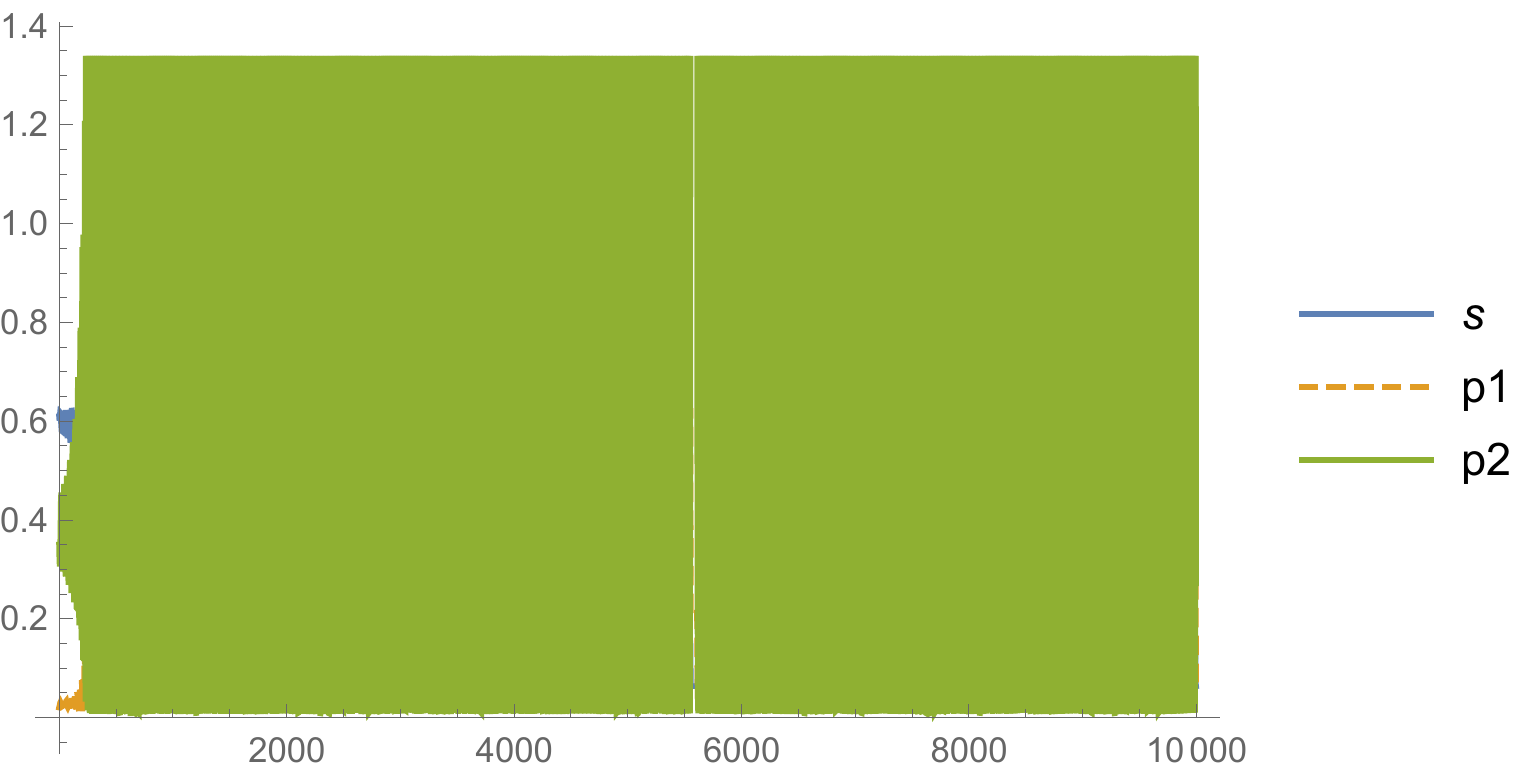}
         \caption{\emph{Time series of all three populations
 }}
     \end{subfigure}
        \caption{\emph{Depiction of the situation at $A_4>A_4^{H}$, which represents the state following a subcritical Hopf bifurcation at $A_4=A_4^{H}=1.2252$.}
}
        \label{after A4}
\end{figure}
\vspace*{1cm}
 On the further increase of the parameter $A_4$(see\emph{ figure (\ref{A4})}), at $A_4^{H*}$=12.071019, the system (\ref{c}) undergoes another Hopf bifurcation at the equilibrium point , which is supercritical, as it satisfies the NASC , $G_1(A_4^{H*})=0.0851308>0 $, $G_2(A_4^{H*})=11.2832>0 $, $G_3(A_4^{H*})=0.960551>0$ , $G_1(A_4^{H*})$ $G_2(A_4^{H*})$-$G_3(A_4^{H*})$=0, and 
$G_1(A_4^{H*})$ $G_2'$($A_4^{H*}$) +$G_2(A_4^{H*})$ $G_1'$($A_4^{H*}$) -$G_3'$($A_4^{H*}$) = 4.78344 $ \neq 0$. Also we get, $\alpha^{'}(0)<0$, $\mu_2<0$ and $\beta_2<0$. Therefore, according to theorem (\ref{t10}) , it indicates the supercritical nature of the Hopf bifurcation. Also, the value of the First Lyapunov coefficient = -7.746701. Figures (\ref{before a4s}) and \ref{after a4s}) describe the exact situation. The coexisting equilibrium point ceases to exist after an increment of the parameter value $A_4=A_4^S=12.248$, which is the saddle-node bifurcation point. When $A_4=3$ the second predator species faces extinction whereas other species survive as $A_1 A_3=0.04<A_4$ , $\frac{A_4 A_6-A_1 A_3 A_6}{A_1 A_3}=3.848 < A_5$ , $\frac{A_1 A_3 A_4 A_5+A_1A_3 A_4 A_6-A_4^2 A_6-A_4^2 A_7}{A_1^2 A_3^2 -A_1 A_3 A_4}=344.86 > A_8$, and $  \frac{(A_1 A_3 A_5 + A_1 A_3 A_6 - A_4 A_6)}{A_4} = 0.00917333 < A_7$.\\
    
   \begin{figure}[H]
     \centering
     \begin{subfigure}[F]{0.45\textwidth}
         \centering
         \includegraphics[width=1.2\textwidth]{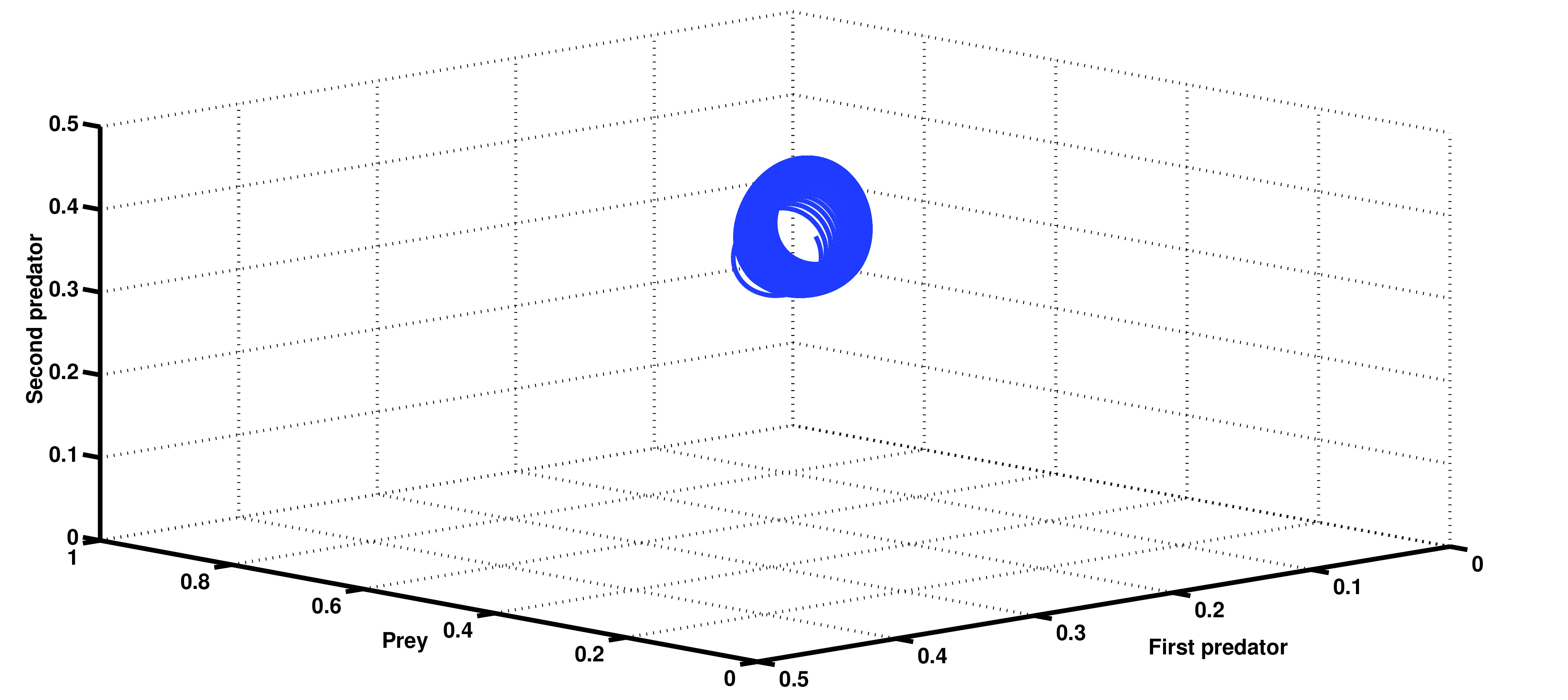}
         \caption{\emph{Phase portrait }}
         \label{ph5}
     \end{subfigure}
     \hfill
     \begin{subfigure}[F]{0.45\textwidth}
         \centering
         \includegraphics[width=\textwidth]{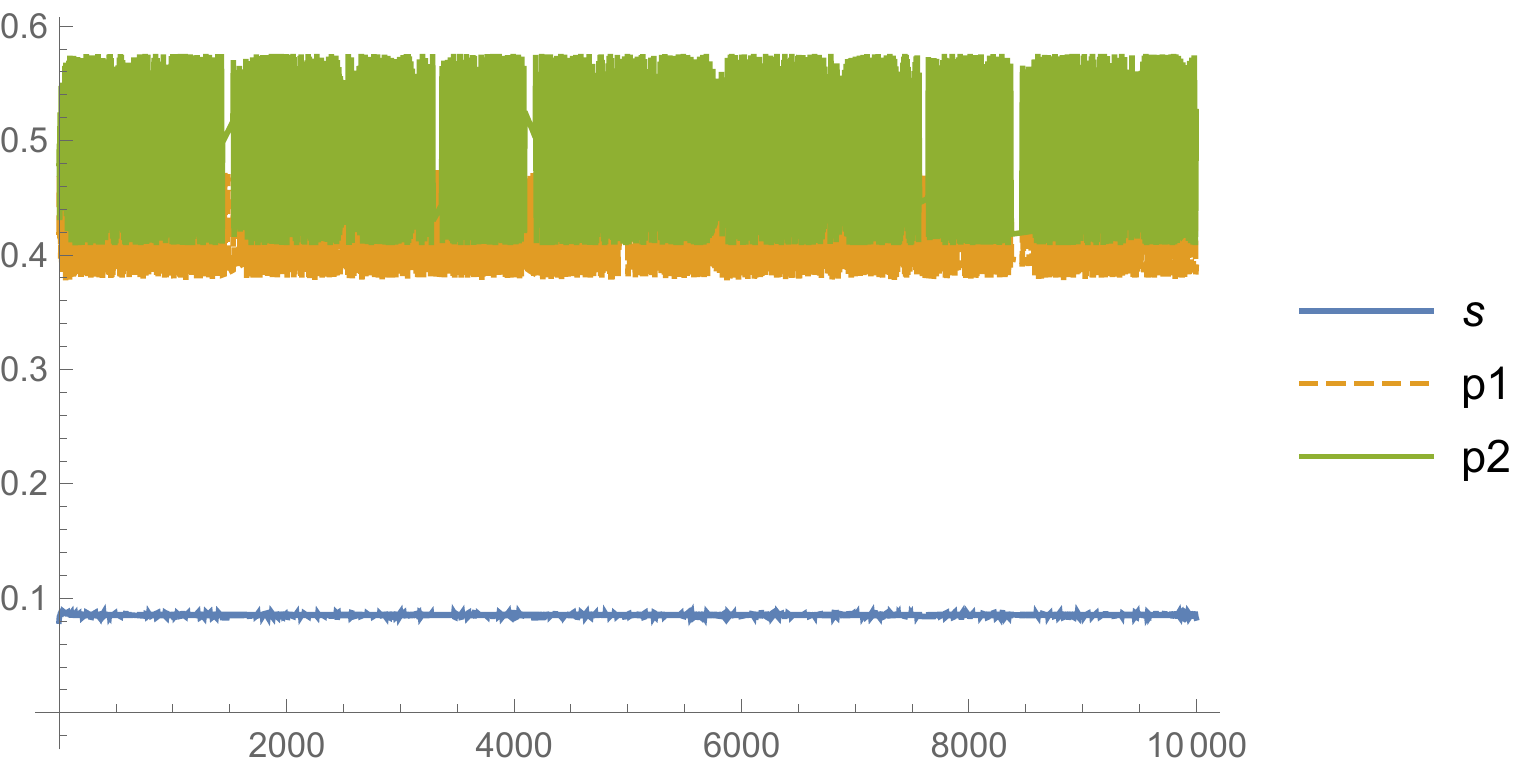}
         \caption{\emph{Time series of all three populations}}
     \end{subfigure}
      \hfill
     \begin{subfigure}[F]{0.45\textwidth}
         \centering
         \includegraphics[width=\textwidth]{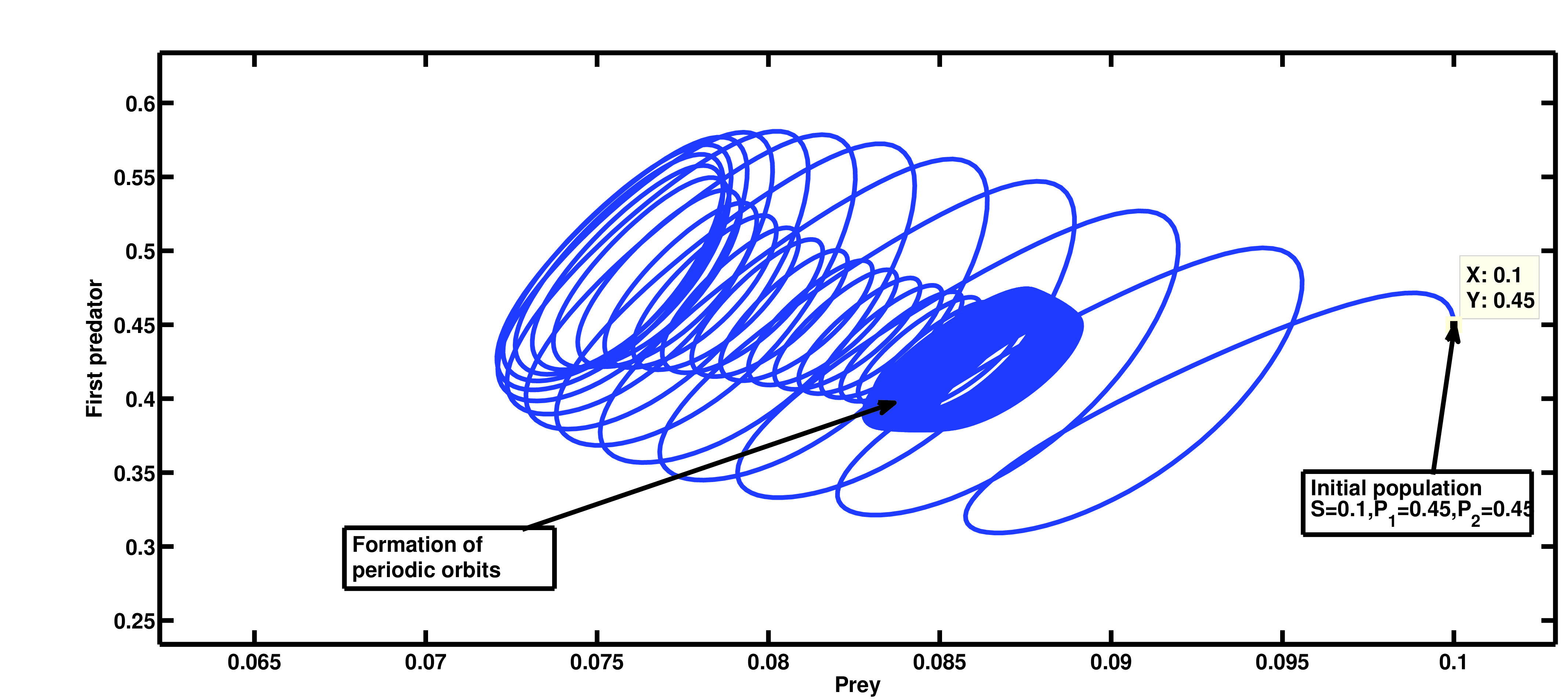}
         \caption{\emph{Prey vs first predator}}
     \end{subfigure}
      \hfill
     \begin{subfigure}[F]{0.45\textwidth}
         \centering
         \includegraphics[width=\textwidth]{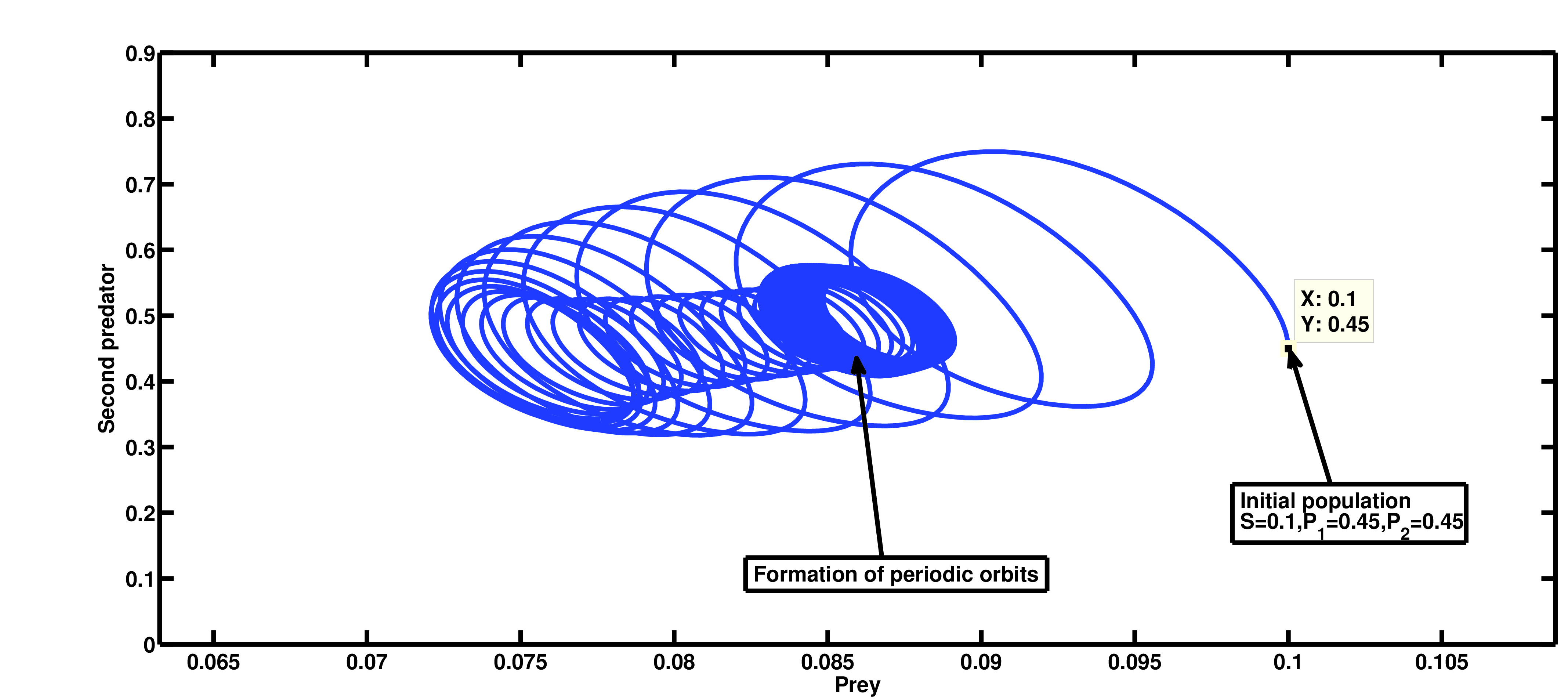}
         \caption{\emph{Prey vs second predator}}
     \end{subfigure}
      \hfill
     \begin{subfigure}[F]{0.45\textwidth}
         \centering
         \includegraphics[width=\textwidth]{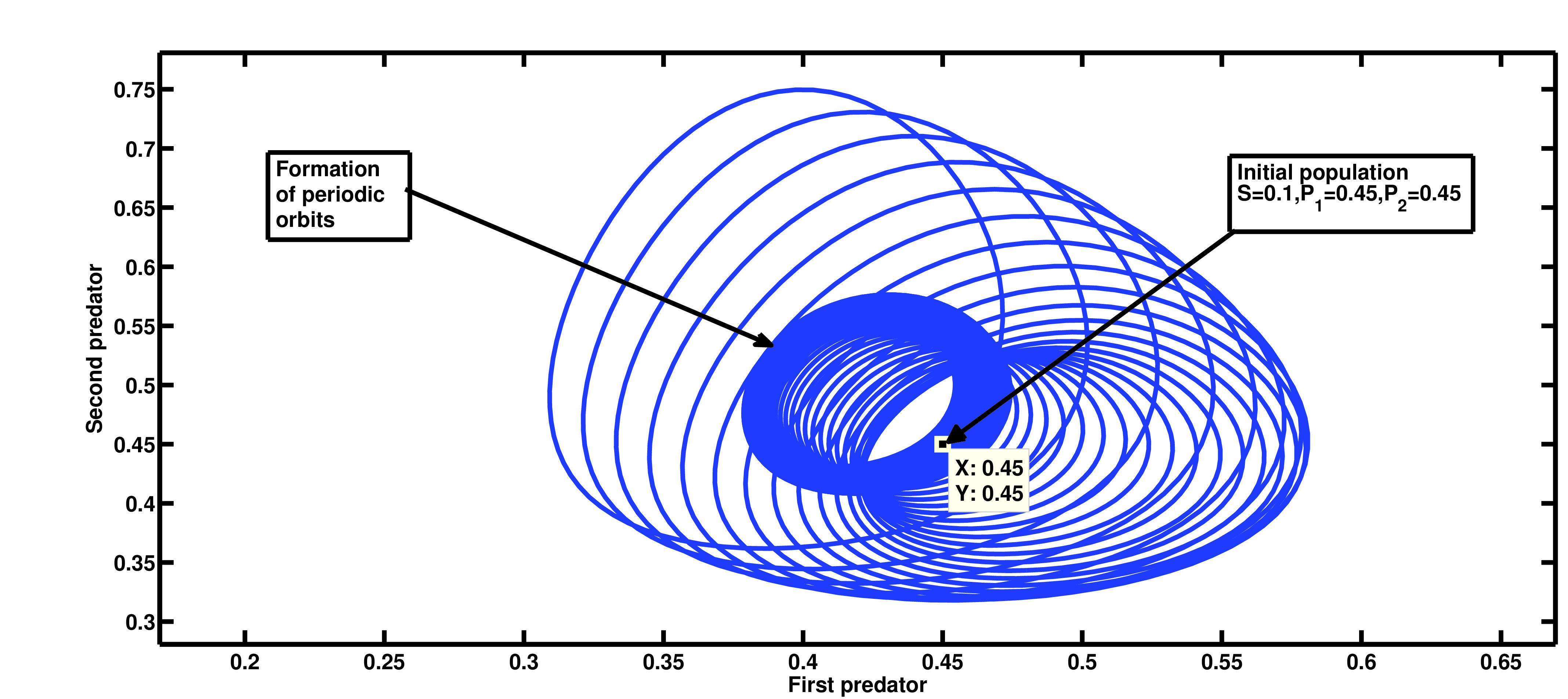}
         \caption{\emph{First predator vs second predator}}
     \end{subfigure}
        \caption{\emph{Illustration of the scenario at $A_4<A_4^{H*}$ (i.e., the state before super critical Hopf bifurcation at $A_4=A_4^{H*}=12.071019$.)}
}
        \label{before a4s}
\end{figure}
\vspace*{1cm}

\begin{figure}[H]
     \centering
     \begin{subfigure}[F]{0.45\textwidth}
         \centering
         \includegraphics[width=1.2\textwidth]{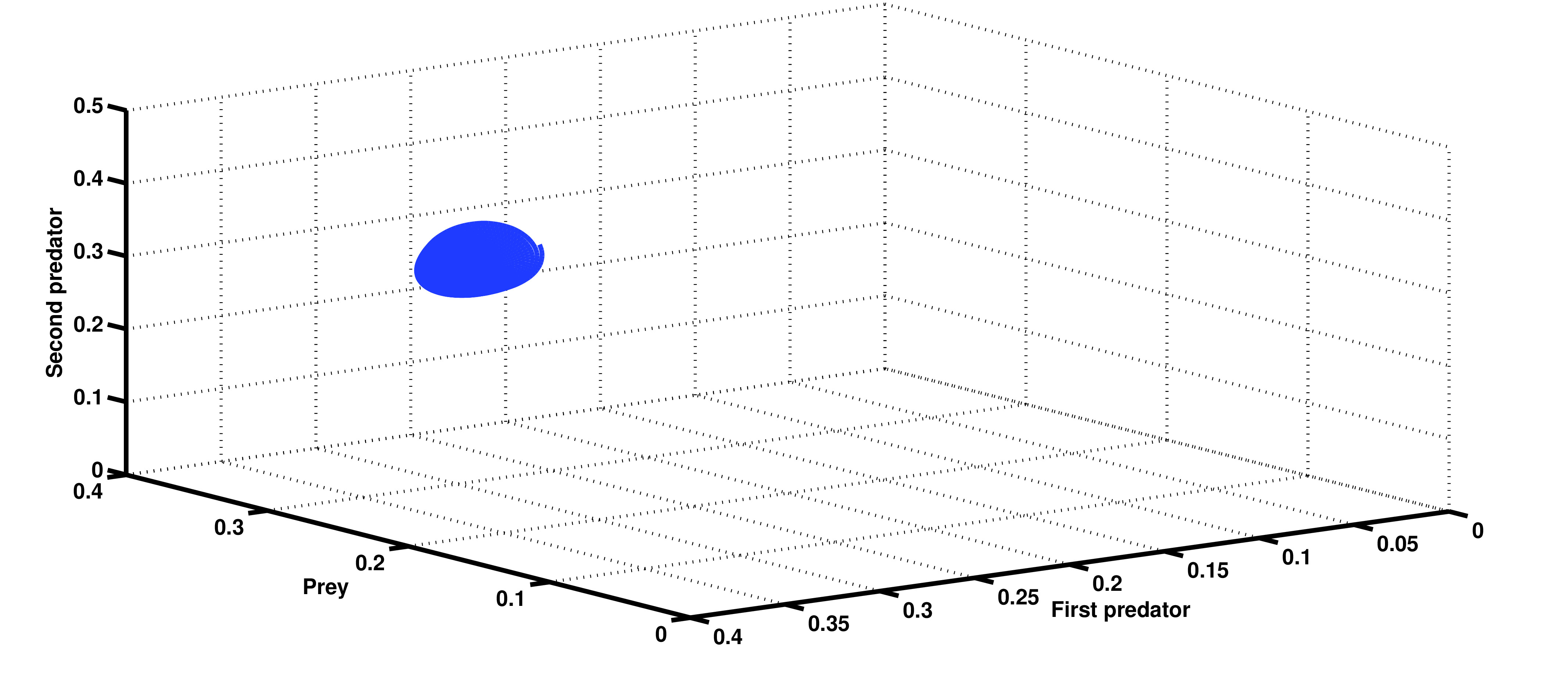} 
         \caption{\emph{Phase portrait} }
         \label{ph4}
     \end{subfigure}
     \hfill
     \begin{subfigure}[F]{0.45\textwidth}
         \centering
         \includegraphics[width=\textwidth]{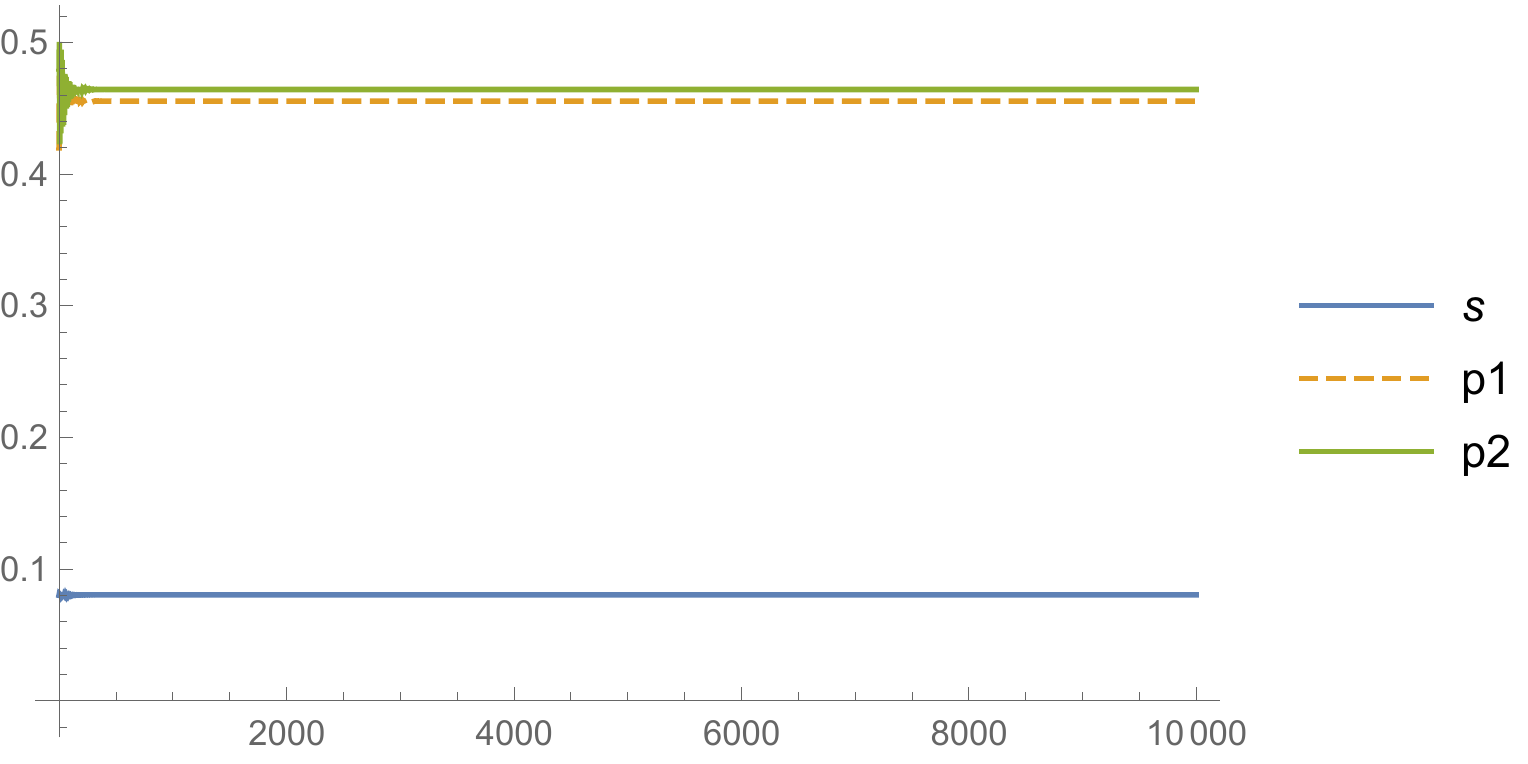}
         \caption{\emph{Time series of all three populations}}
     \end{subfigure}
      \hfill
     \begin{subfigure}[F]{0.45\textwidth}
         \centering
         \includegraphics[width=\textwidth]{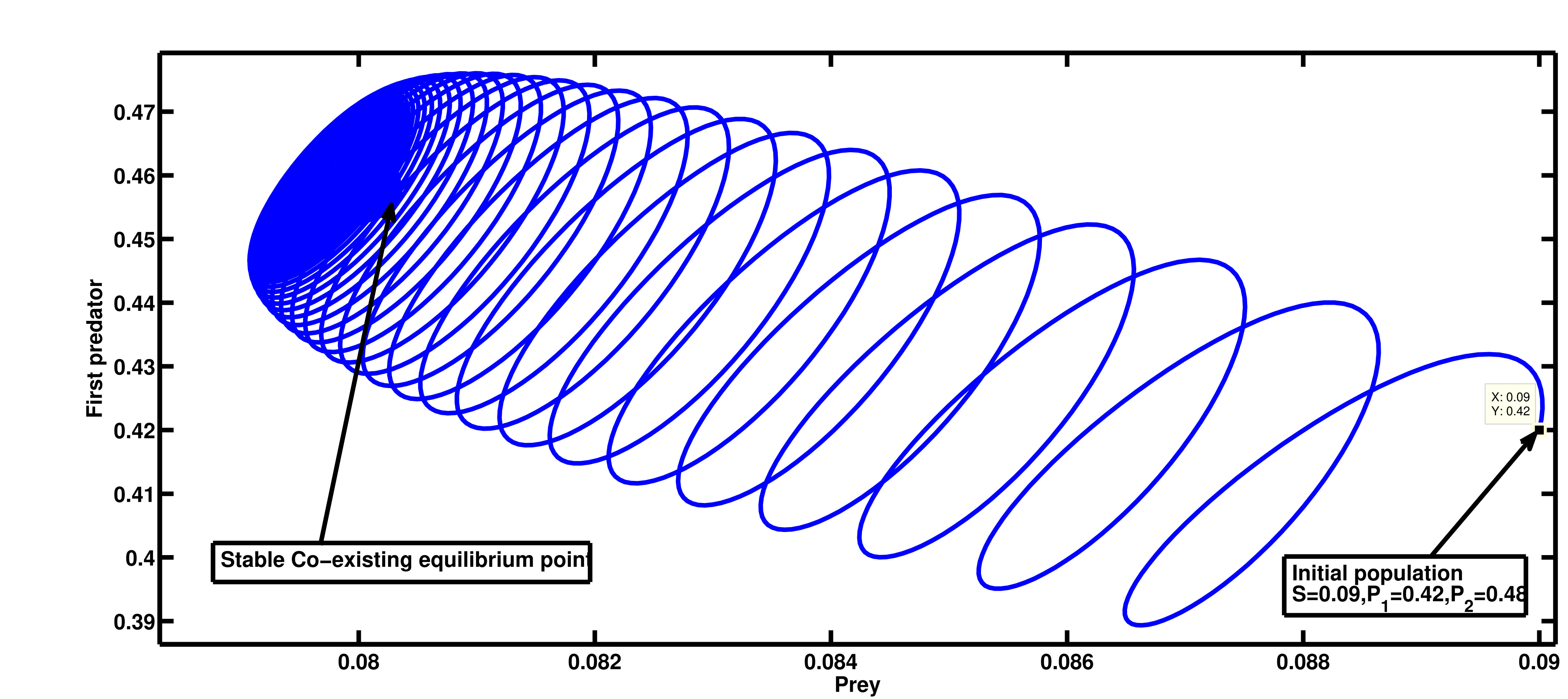}
         \caption{\emph{Prey vs first predator}}
     \end{subfigure}
      \hfill
     \begin{subfigure}[F]{0.45\textwidth}
         \centering
         \includegraphics[width=\textwidth]{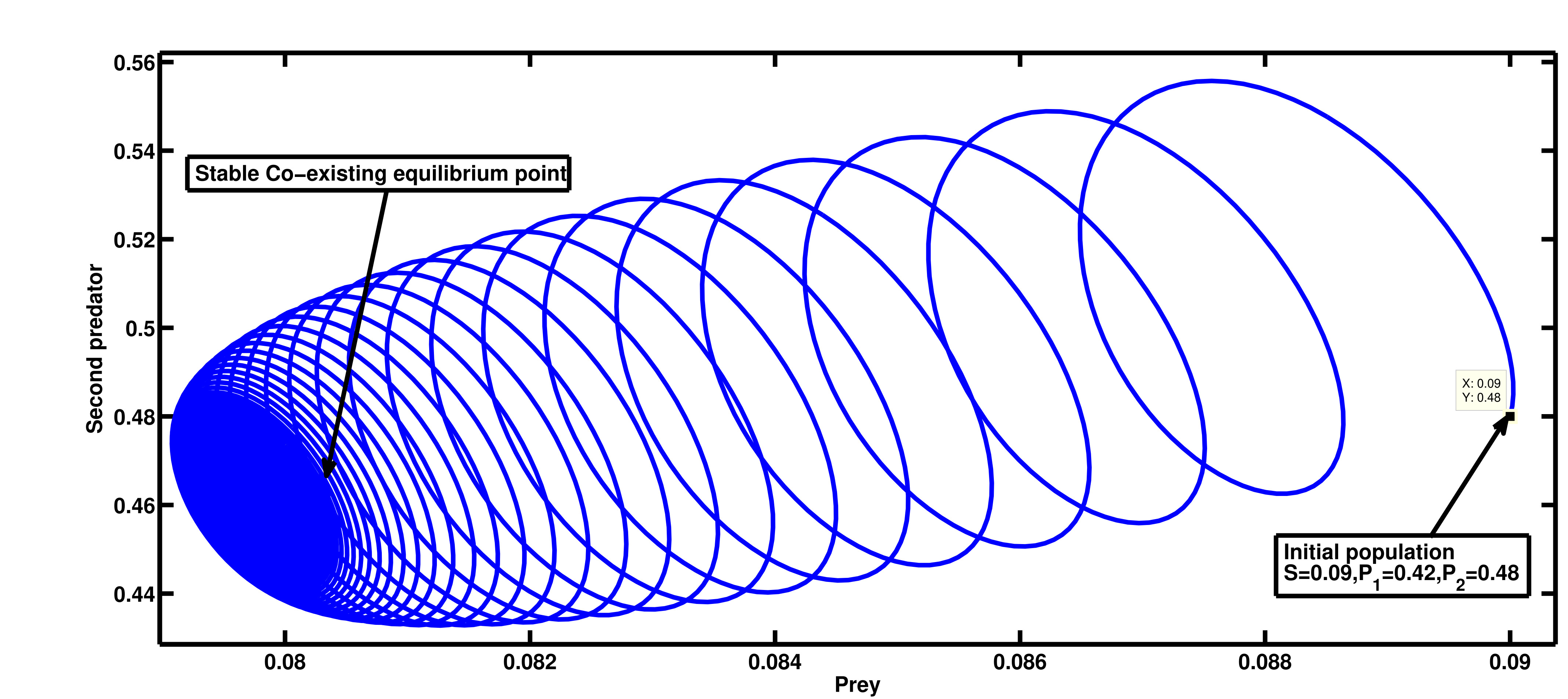}
         \caption{\emph{Prey vs second predator}}
     \end{subfigure}
      \hfill
     \begin{subfigure}[F]{0.45\textwidth}
         \centering
         \includegraphics[width=\textwidth]{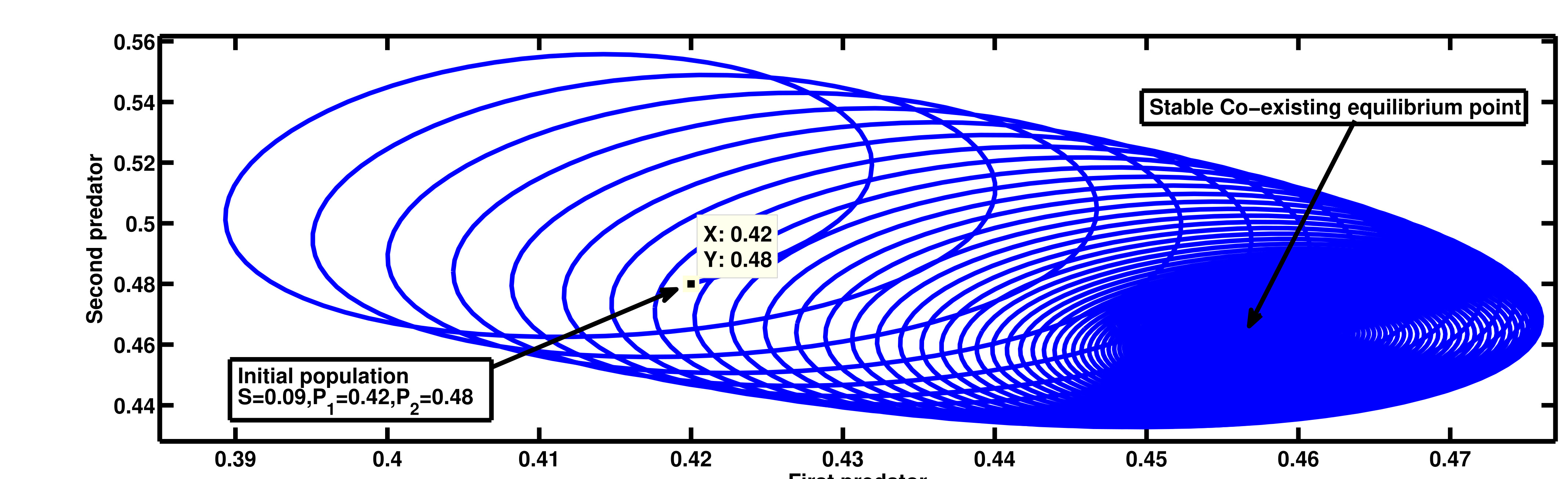}
         \caption{\emph{First predator vs second predator}}
     \end{subfigure}
        \caption{\emph{An illustration of the scenario at $A_4>A_4^{H*}$, which represents the state following a supercritical Hopf bifurcation at $A_4=A_4^{H*}=12.071019$.}
}
        \label{after a4s}
\end{figure}
 
Figure (\ref{A5}) depicts the effect of $A_5$ ( related to the conversion efficiency or growth rate of the second predator) on the coexisting equilibrium point as the said parameter is varied keeping rest of the parameter values as same as in table  (\ref{table:1}). As $A_5$ is increased the unstable branch of the coexisting equilibrium point becomes stable through supercritical Hopf bifurcation occurring at $A_5=A_5^{H*}=2.007197$, and the equilibrium point gains stability as all the NASC of existence of Hopf bifurcation (as discussed in Theorem \ref{t9}) are satisfied as $G_1(A_5^{H*})=0.825096>0 $, $G_2(A_5^{H*})=5.83388>0 $, $G_3(A_5^{H*})=4.81351>0$ , $G_1(A_5^{H*})$ $G_2(A_5^{H*})$-$G_3(A_5^{H*})$=0, and $G_1(A_5^{H*})$ $G_2'$($A_5^{H*}$) +$G_2(A_5^{H*})$ $G_1'$($A_5^{H*}$) -$G_3'$($A_5^{H*}$) = -5.83715 $ \neq 0.$ and we get,  $\alpha^{'}(0)<0$, $\mu_2<0$ and $\beta_2<0$. Also, first Lyapunov coefficient = $-4.606137e^{-1}$. For example, at $A_5=4.536$, $G_1=0.844371>0$, $G_2=1.96046>0$, $G_3=1.22722>0$, and $G_1 G_2-G_3= 0.428135>0$ i.e., coexisting equilibrium is stable. The situation is accurately depicted by Figures (\ref{before a5s}) and (\ref{after a5s}). However, through transcritical bifurcation which occurs at $A_5=A_5^{T}=5.3449$, the stable branch becomes unstable again, but at the same time another branch of the first predator free equilibrium point occurs with comparatively low prey density and gains stability. For instance, at $A_5=5.5$, the extinction of the first predator but the survival of the other two species is observed as $A_5=5.5>A_7$ and $\frac{A_2 A_5^2-2A_2A_5A_7+A_2A_7^2+A_3A_5^2-A_3A_5A_7}{A_5 A_7}=0.42153>A_4$.
   \begin{figure}[H]
         \centering
         \includegraphics[width=\textwidth]{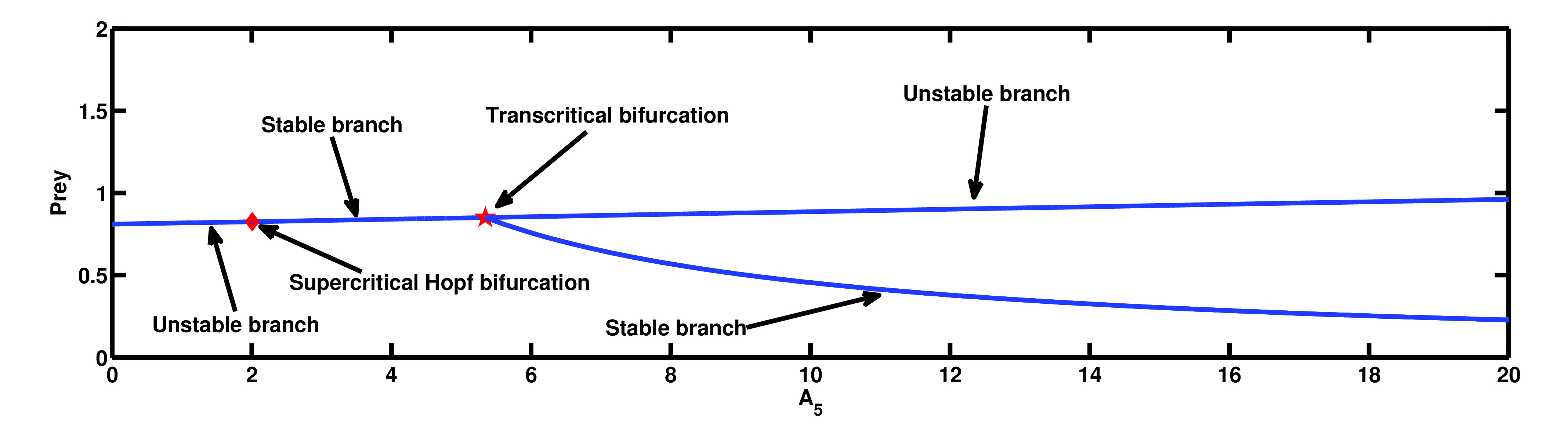}
         \caption{\emph{The occurrence of supercritical Hopf bifurcation and transcritical bifurcation for parameter $A_5$}}
         \label{A5}
     \end{figure} 
     
  
  The parameter related to the mortality rate of the second predator is denoted by $A_7$. Figure (\ref{A7}) illustrates the dynamical scenario of the system as $A_7$ is varied keeping the rest of the parameter values as in table  (\ref{table:1}). At $A_7=30$, the second predator becomes extinct as  $A_1 A_3=0.04<A_4$ , $\frac{A_4 A_6-A_1 A_3 A_6}{A_1 A_3}=0.468 < A_5$ , $\frac{A_1 A_3 A_4 A_5+A_1A_3 A_4 A_6-A_4^2 A_6-A_4^2 A_7}{A_1^2 A_3^2 -A_1 A_3 A_4}=328.813 > A_8$ , and $ \frac{(A_1 A_3 A_5 + A_1 A_3 A_6 - A_4 A_6)}{A_4} = 0.4068 < A_7$. 
  However, at $A_7=26.869$, two interior equilibrium points appear through saddle-node bifurcation, of which the one with higher prey density gains stability, but as the said parameter value is decreased further, there is a transcritical bifurcation at $A_7=A_7^T=3.857$, which exchanges the stability of the said interior equilibrium point to first predator free equilibrium point having comparatively low prey density.,For example, reducing the death rate of the second predator to $A_7=1$,  causes the first predator to become extinct as $A_7=1<A_5$ and $\frac{A_2 A_5^2-2A_2A_5A_7+A_2A_7^2+A_3A_5^2-A_3A_5A_7}{A_5 A_7}=7.20982>A_4$ showing stability of the first predator-free equilibrium .

   \begin{figure}[H]
         \centering
         \includegraphics[width=\textwidth]{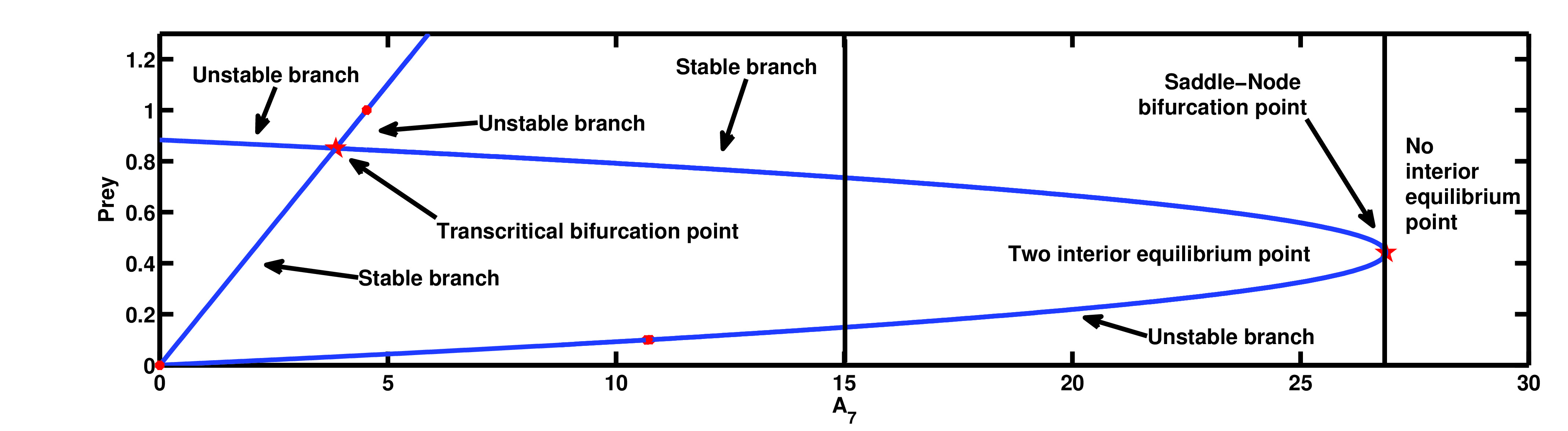}
         \caption{\emph{The manifestation of transcritical bifurcation and saddle-node bifurcation can be observed for parameter $A_7$}}
         \label{A7}
     \end{figure} 
\subsubsection{\emph{Role of kleptoparasitism in changing the qualitative scenario}}
In the biosystem (\ref{b}), $a$ is the parameter related to the rate of kleptoparasitism. When $a \rightarrow \infty $ then $A_1 \rightarrow 0$ , $A_4 \rightarrow 0$, and $\lim_{a\to\infty}\frac{A_4}{A_1+P_2}=\lim_{a\to\infty}\frac{r_2 r_4 k}{(1+a)r_1}=0$ i.e., no growth of the first predator. In other words, when the rate of kleptoparasitism by the second predator is exceptionally high, the growth rate of the first predator declines to zero, and the species eventually disappears from the system. When $a\rightarrow 0$, then $A_1 \rightarrow \infty$ , $A_4 \rightarrow \infty$ and $\lim_{a\to\infty}\frac{A_4}{A_1+P_2}=\lim_{a\to\infty}\frac{r_2 r_4 k}{(1+a)r_1}=\frac{r_2 r_4 k}{r_1}$ i.e., when the kleptoparasitism rate is negligible, the two predators can coexist.

 \begin{figure}[H]
         \centering
         \includegraphics[width=\textwidth]{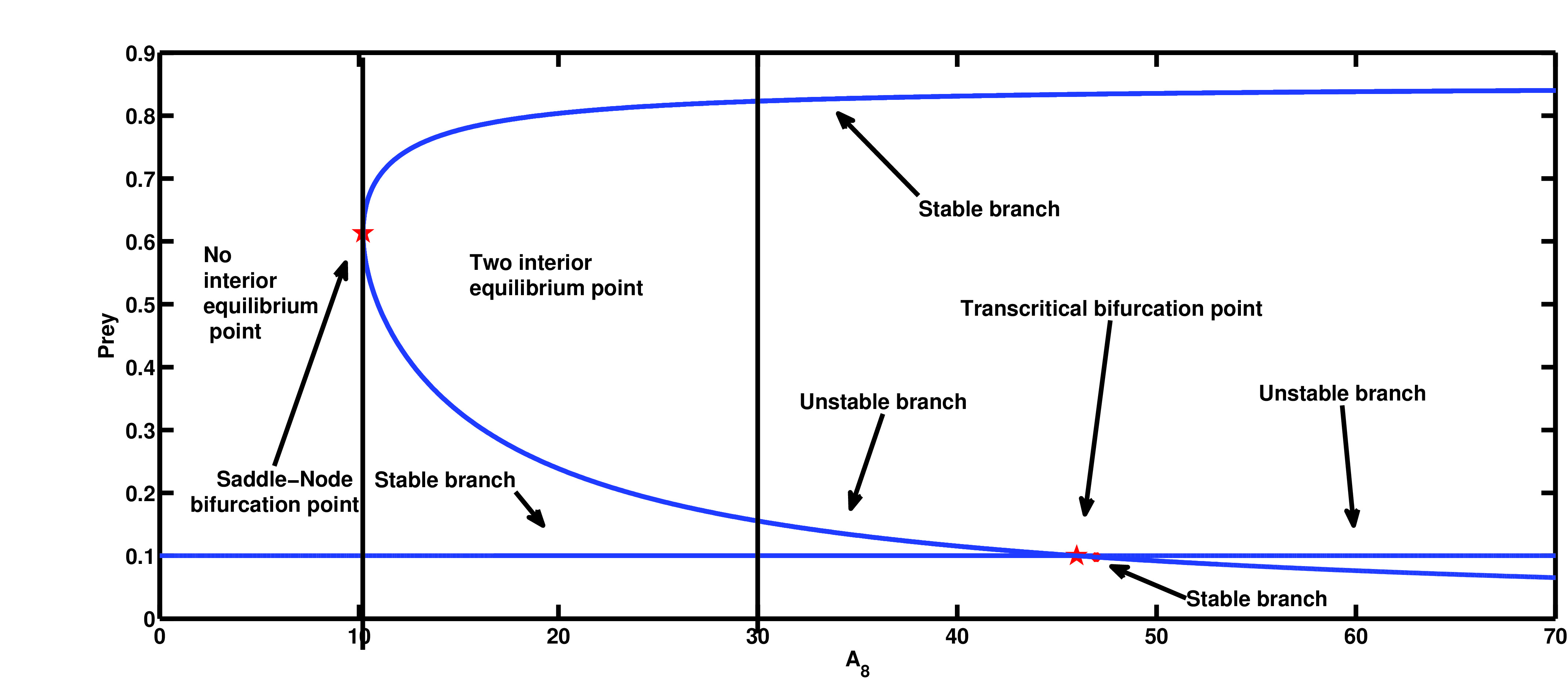}
         \caption{\emph{The equilibrium curves of two interior equilibrium alongwith the second predator-free equilibrium. It also depicts the emergence of saddle-node bifurcation and transcritical bifurcation with respect to parameter $A_8$}}
         \label{A8}
     \end{figure}




$A_8$ is the parameter relating to the conversion rate of the second predator of the food received through kleptoparasitism. When $A_8=0$ and the rest parameter values from table  (\ref{table:1}), then the second predator population continues to decline and disappears from the system after a certain period of time, whereas the first predator and prey species persist having a population density $E_2(0.1,0.9,0)$ through time. This is because $A_1 A_3=0.04<A_4$ , $\frac{A_4 A_6-A_1 A_3 A_6}{A_1 A_3}=0.468 < A_5$ , $\frac{A_1 A_3 A_4 A_5+A_1A_3 A_4 A_6-A_4^2 A_6-A_4^2 A_7}{A_1^2 A_3^2 -A_1 A_3 A_4}=45.9911 > A_8$ , and $  \frac{(A_1 A_3 A_5 + A_1 A_3 A_6 - A_4 A_6)}{A_4} = 0.4068 < A_7$. As $A_8$ is increased, the birth of interior equilibrium occurs through saddle-node bifurcation at the bifurcation parameter value $A_8=A_8^S=10.2$. One of the equilibrium points becomes stable, while the other one is unstable. The equilibrium point $E_2$ still maintains stability making the system bistable between second predator-free equilibrium and interior equilibrium. At $A_8=45.9726=A^T$, a transcritical bifuraction takes place  changing the stability of the coexisting equilibrium point and the second predator-free equilibrium point, and the system becomes bistable due to the local stability behaviour of the two interior equilibrium points.  


\par

\begin{figure}[H]
     \centering
     \begin{subfigure}[F]{0.4\textwidth}
         \centering
         \includegraphics[width=1.2\textwidth]{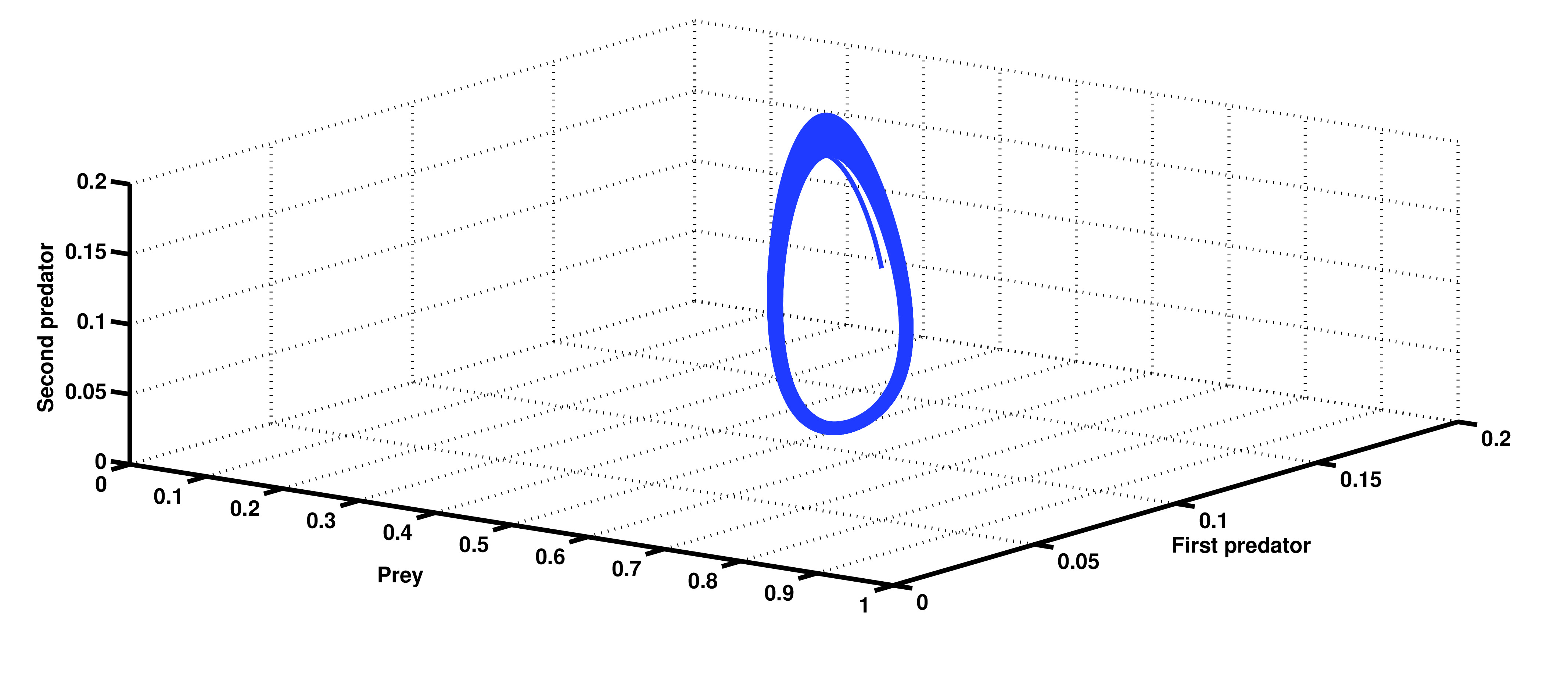}
         \caption{\emph{Phase portrait }}
         \label{ph3}
     \end{subfigure}
     \hfill
     \begin{subfigure}[F]{0.5\textwidth}
         \centering
         \includegraphics[width=\textwidth]{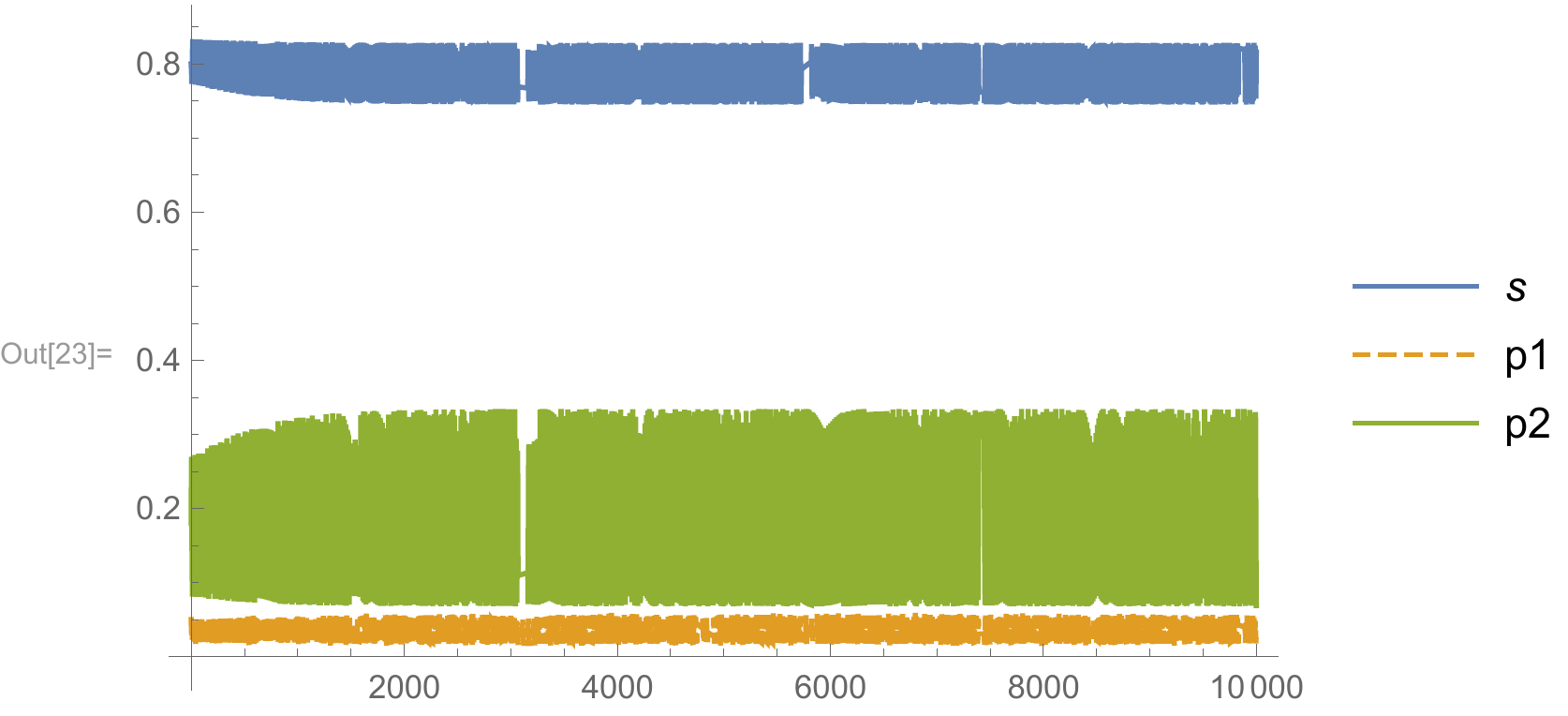}
         \caption{\emph{Time series of all the three populations}}
     \end{subfigure}
        \caption{\emph{An illustration is provided for the scenario occurring at $A_5<A_5^{H*}$, which corresponds to the state prior to the occurrence of a supercritical Hopf bifurcation  at $A_5=A_5^{H*}=2.007197$.
}
}
        \label{before a5s}
\end{figure}
\vspace*{1cm}

\begin{figure}[H]
     \centering
     \begin{subfigure}[F]{0.4\textwidth}
         \centering
         \includegraphics[width=1.2\textwidth]{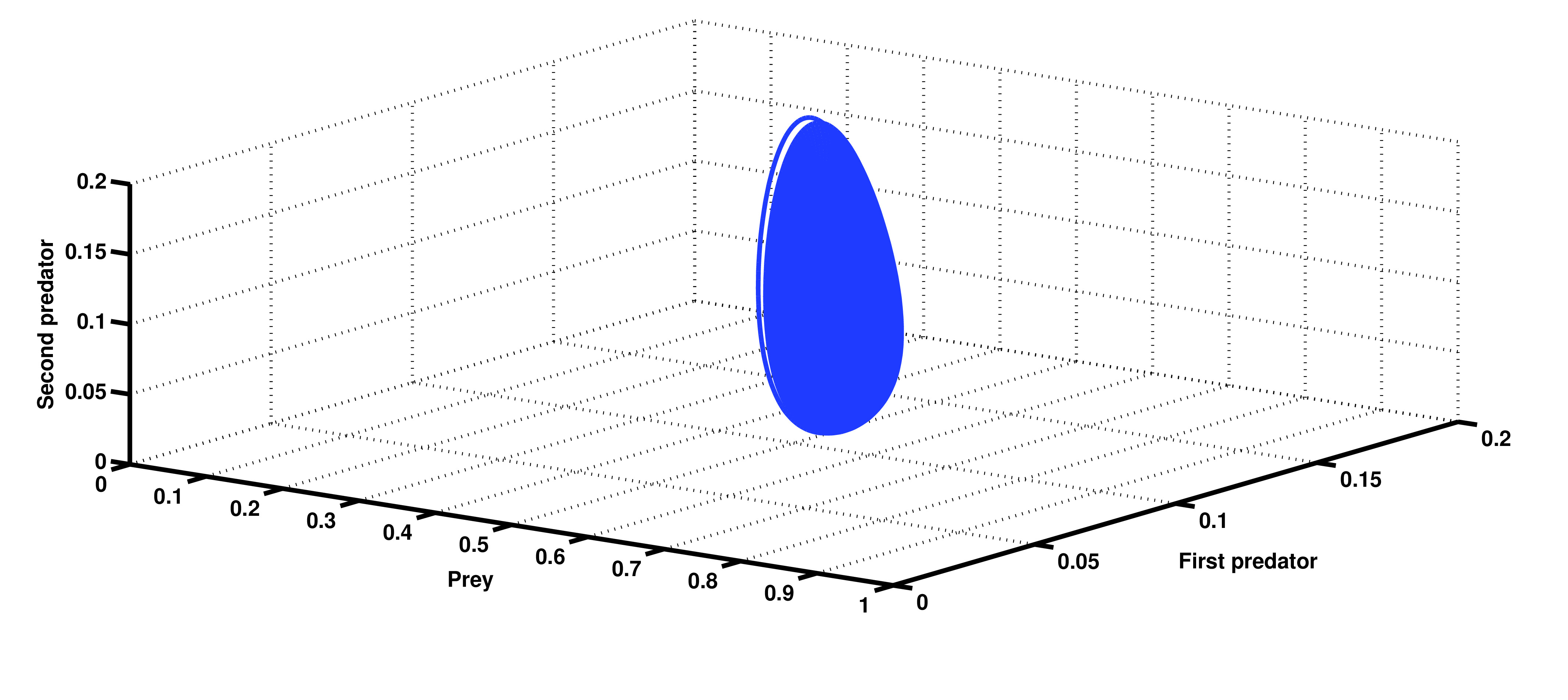}
         \caption{\emph{Phase portrait}}
         \label{ph2}
     \end{subfigure}
     \hfill
     \begin{subfigure}[F]{0.5\textwidth}
         \centering
         \includegraphics[width=\textwidth]{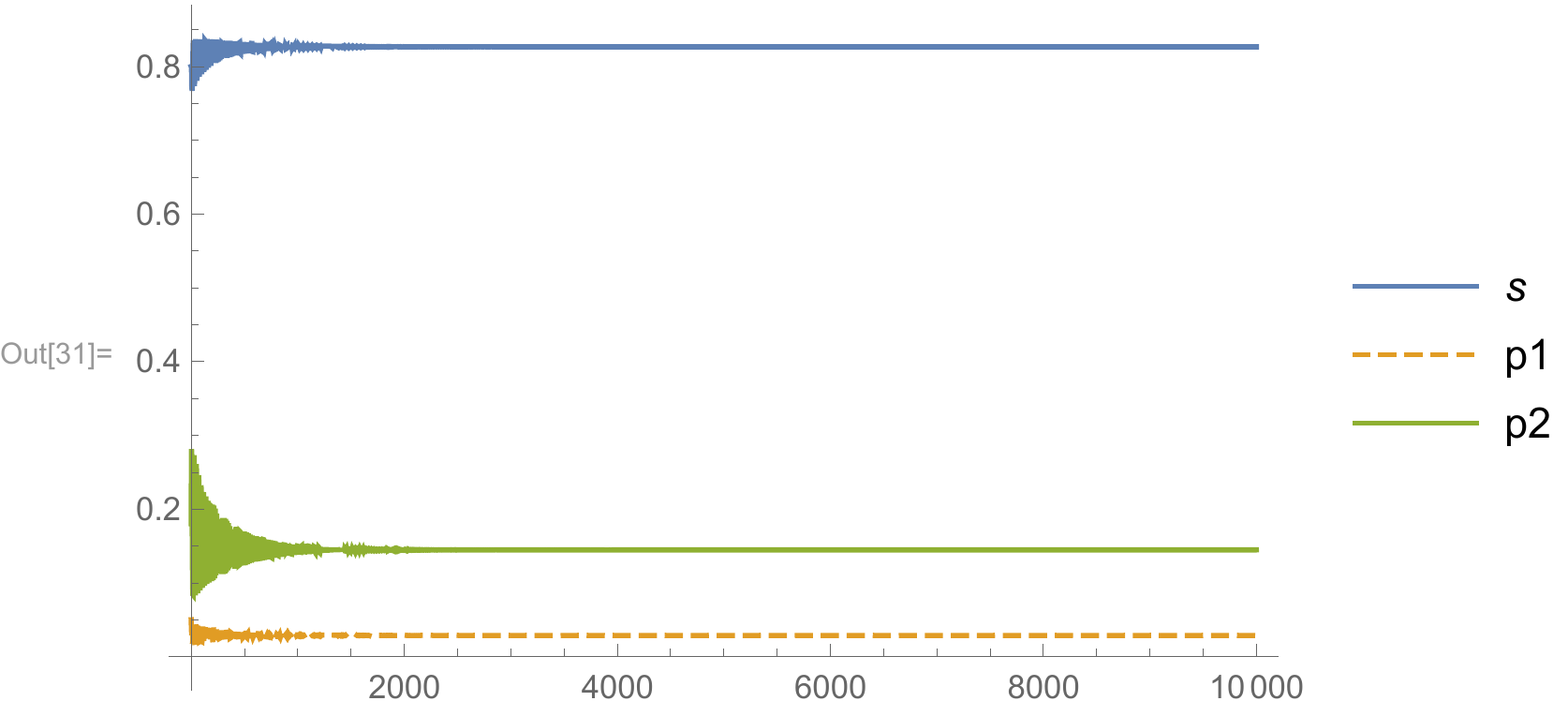}
         \caption{\emph{Time series of all the three populations}}
     \end{subfigure}
        \caption{\emph{A depiction of the scenario occurring at $A_5>A_5^{H*}$, which denotes the state after a supercritical Hopf bifurcation at $A_5=A_5^{H*}=2.007197$
}
}
        \label{after a5s}
\end{figure}
\vspace*{1cm}

\subsection{Numerical simulation of bifurcation scenario of codimension  two}
Figure (\ref{A5A1}) shows the trajectory of the Hopf bifurcation curve when both parameters $A_1$ and $A_5$ are varied simultaneously. As $A_5$ is increased, there occurs one of the codimension two bifurcation, generalized Hopf bifurcation at the parameter value $(A_5,A_1)=(1.0216,0.1367)$, which changes the direction of Hopf bifurcation from super critical to subcritical one. Again, there is another generalized Hopf bifurcation at the parameter value $(A_5,A_1)=(1.4996,0.1045)$ showing a change in the direction of Hopf bifurcation from subcritical to supercritical one. Next, considering different parameters $A_4$ and $A_2$, a Hopf bifurcation curve is drawn in figure(\ref{A4A2}), this curve is strictly increasing as $A_4$ is increased and an existence of one generalized Hopf bifurcation point is seen i.e., there is a change in the direction of Hopf bifurcation from supercritical to subcritical one as $A_4$ is increased. In figure(\ref{A4A5}) , there are two generalized Hopf bifurcation points at the parameter values $(A_5,A_4)=(1.1230,0.2896)$ and $(A_5,A_4)=(3.8733,0.9586)$ along an increasing Hopf bifurcation curve, when both $A_4,A_5$ are varied, and $A_5$ is in the abscissa. Correspondingly, the direction of Hopf bifurcation changes from subritical to supercritical and then again to subcritical one. As $A_3,A_4$ are varied in figure(\ref{A3A4}), the Hopf bifurcation curve again shows two generalized Hopf bifurcation points with the changes in the direction of Hopf bifurcation as explained in figure(\ref{A5A1}).

\begin{figure}[H]
     \centering
     \begin{subfigure}[F]{0.45\textwidth}
         \centering
         \includegraphics[width=\textwidth]{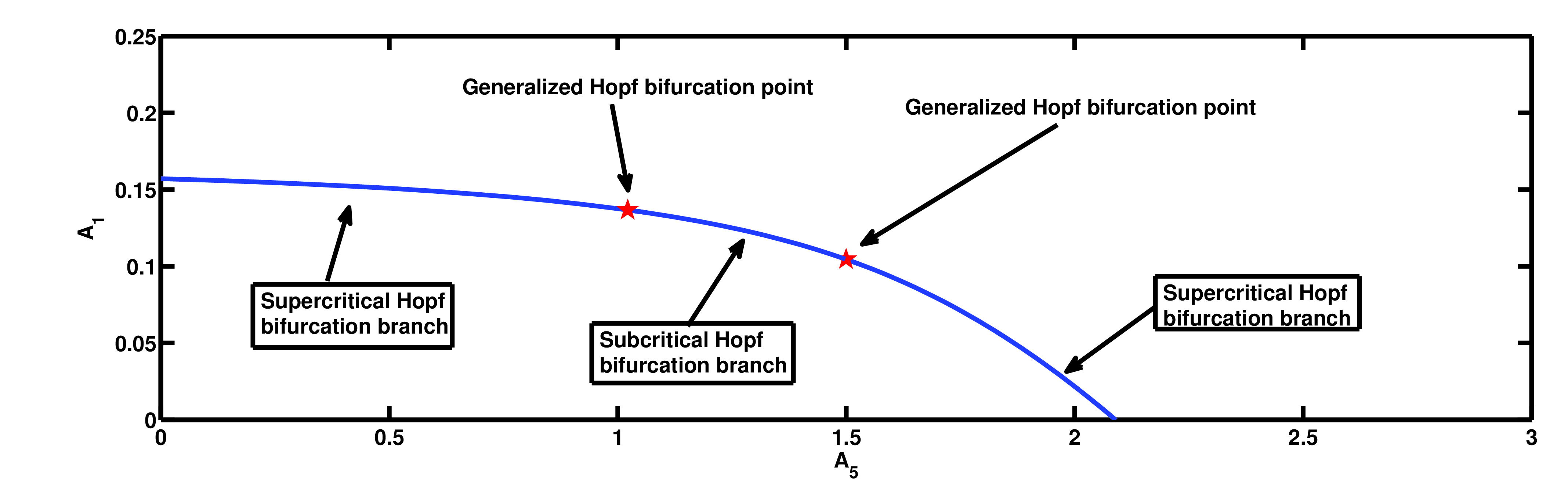}
         \caption{\emph{Hopf bifurcation curve for $A_1$ vs $A_5$}}
          \label{A5A1}
     \end{subfigure}
     \begin{subfigure}[F]{0.45\textwidth}
         \centering
         \includegraphics[width=\textwidth]{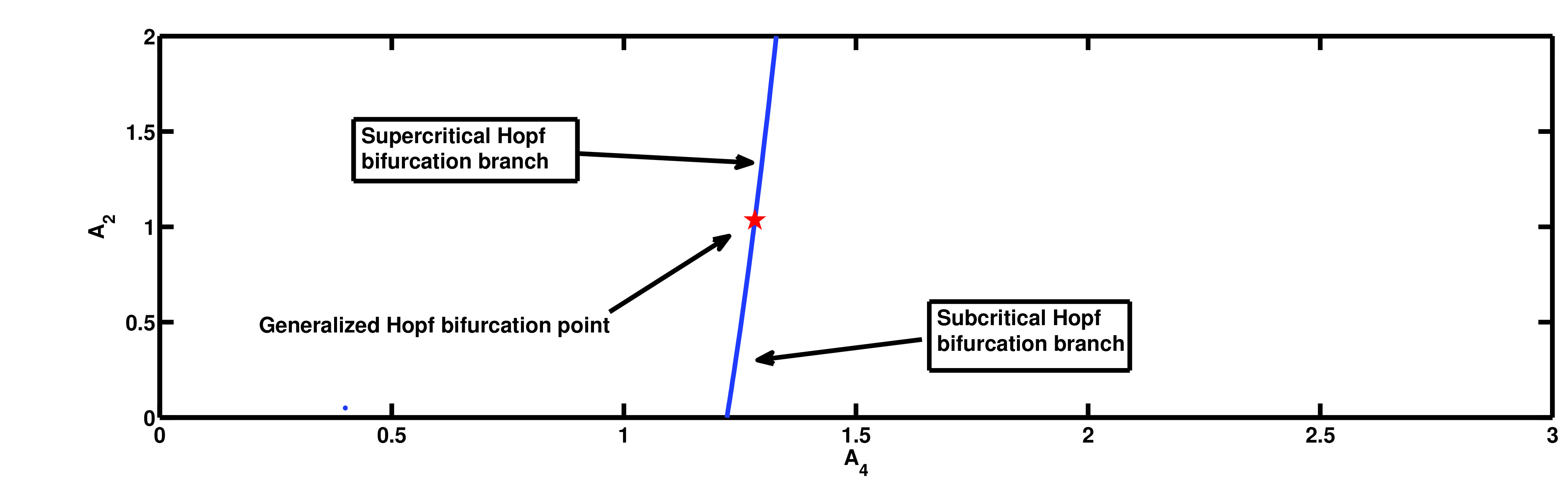}
         \caption{\emph{Hopf bifurcation curve for $A_4$ vs $A_2$}}
          \label{A4A2}
     \end{subfigure}
     \begin{subfigure}[F]{0.45\textwidth}
         \centering
         \includegraphics[width=\textwidth]{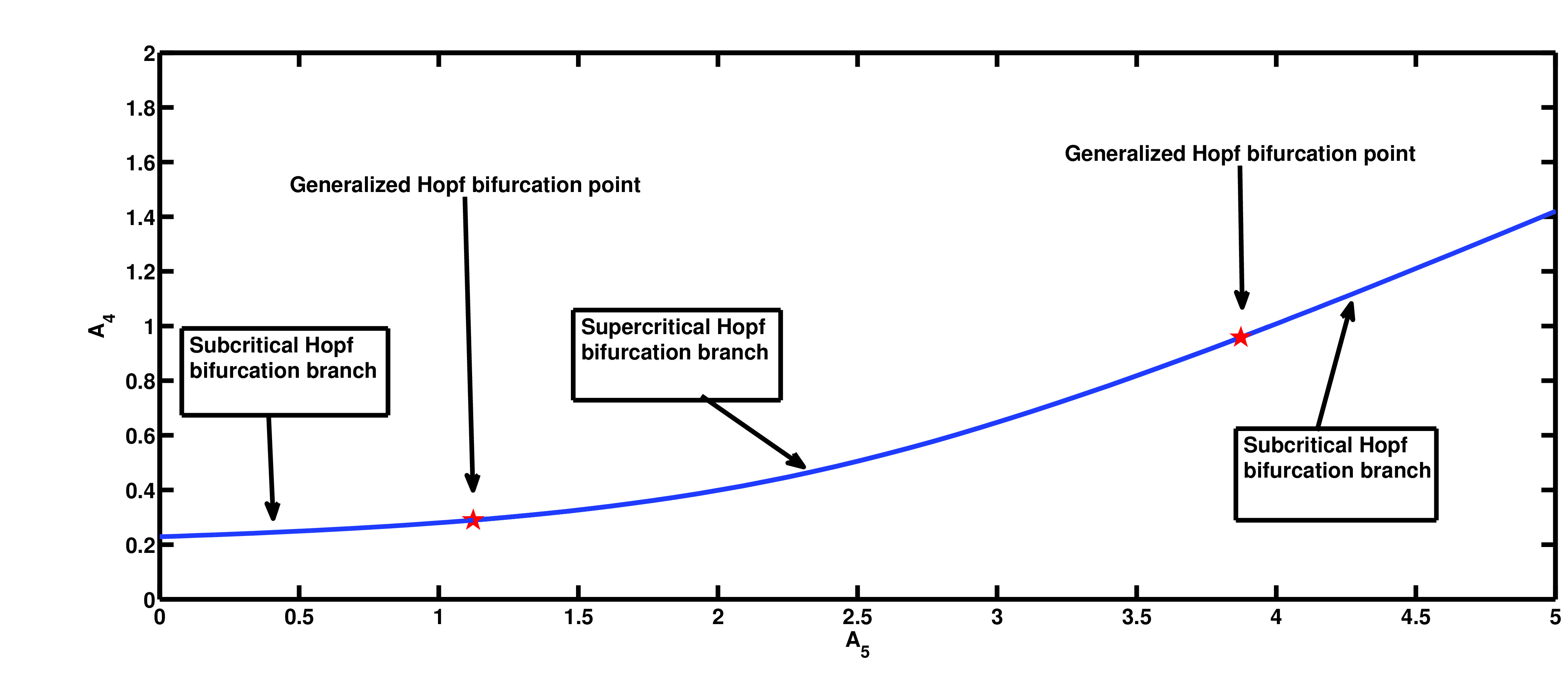}
         \caption{\emph{Hopf bifurcation curve for $A_4$ vs $A_5$}}
          \label{A4A5}
     \end{subfigure}
     \begin{subfigure}[F]{0.45\textwidth}
         \centering
         \includegraphics[width=\textwidth]{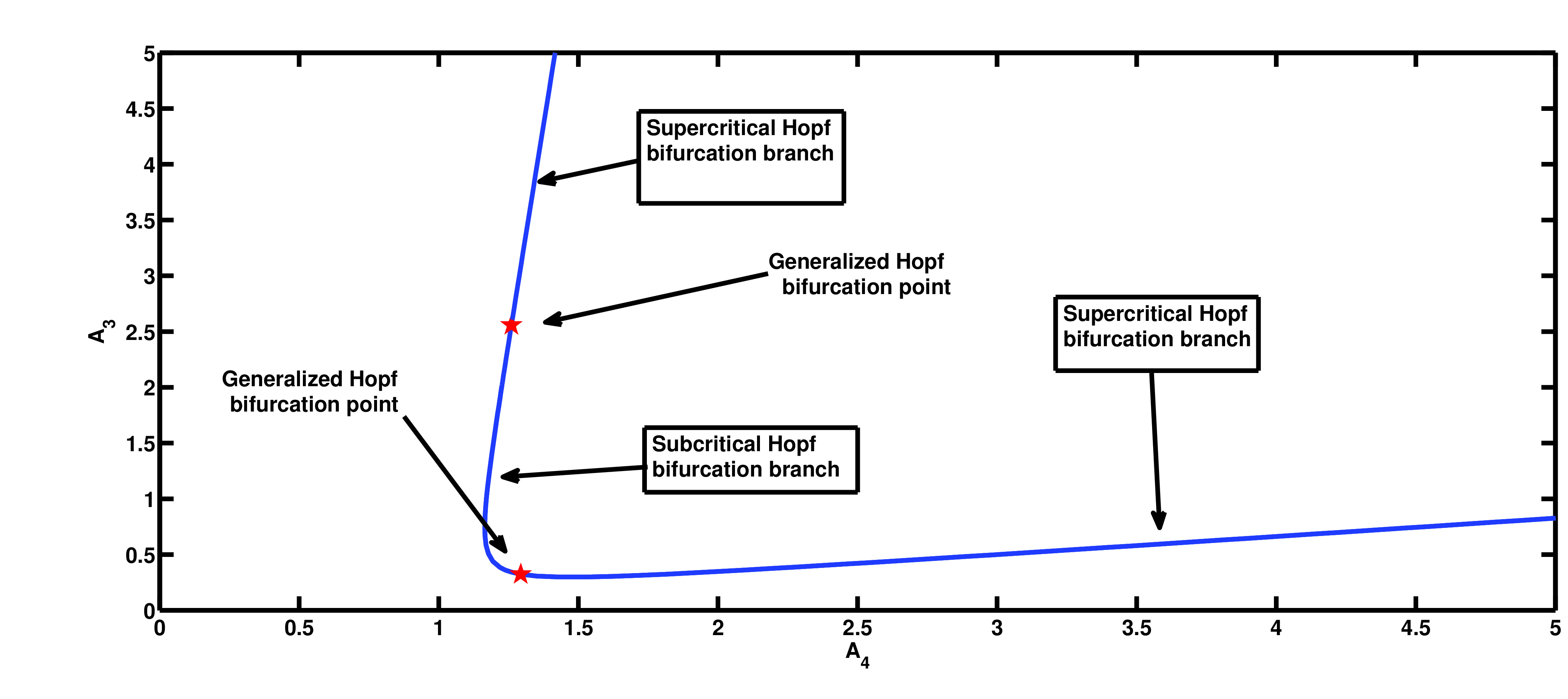}
         \caption{\emph{Hopf bifurcation curve for $A_3$ vs $A_4$}}
          \label{A3A4}
     \end{subfigure}
     
        \caption{\emph{Emergence of generalized Hopf bifurcations}
}
        \label{GH}
\end{figure}
\vspace*{1cm}

\section{Conclusion}
We provided a model of two predators and one prey in the present paper, with Holling type-I functional responses for both predators. 
The two predator species engage in normal interspecific competition. 
Through kleptoparasitism, the first predator furnished the second predator with a supplementary food supply. 
Understanding the role of kleptoparasitism in influencing the population dynamics of the second predator as well as the first predator is one of the most intriguing aspects of this work. Obviously, the second predator is reliant on the first predator for its secondary food source. This paper extracts the dependency scenario of the second predator on the first predator. In this study, the secondary food amount is the result of a post-kelptoparasitism scenario among the predators, which has received little attention in the literature.

\par The biosystem (\ref{c}) demonstrates the possibility of five different types of ecologically feasible equilibrium points. 
A trivial equilibrium point that represents an absence of all species in the system is not the ideal situation from an ecological perspective. 
As both predators are specialist predators, the probable realisation of trivial equilibrium is the overeating of the prey by either or both predators, resulting in the extinction of the prey first, followed by the extinction of both predators owing to a lack of food source. 
Because the trivial equilibrium point is unstable, the long-term result most likely involves the continuous presence of prey species.
Further, if both predators become extinct, the biosystem allows prey species to flourish up to the environmental carrying capacity. This is due to the stability of the axial equilibrium under certain parametric conditions.
Stability conditions for the axial equilibrium point ecologically interpret that the likelihood of reaching the axial equilibrium point increases as the ratio of the first predator's growth rate to its death rate approaches a smaller value and as the second predator's death rate exceeds its growth rate. This is equally true from an ecological standpoint. Surprisingly, interspecific competition between the two predator species plays no role in driving a population trajectory to axial equilibrium.

\par The first boundary equilibrium point depicts a second predator-free system in which the first predator and prey survive but the second predator disappears. One of the necessary requirements for the stability of this equilibrium point is that the first predator's growth rate must be greater than its death rate. It can be seen from the stability conditions of this equilibrium point that the parameter relating to the second predator's loss in growth rate due to competition has a direct impact on this equilibrium's stability. However, it is interesting to note that the loss of the first predator due to competition has little influence on the second predator free equilibrium's stability.


\par The second boundary equilibrium point depicts the first predator-free scenario, in which the second predator and prey species survive but the first predator becomes extinct. The growth rate of the second predator owing to the primary food resource should be greater than its own mortality rate, as the secondary growth rate due to kleptoparasitism will be ineffective due to the extinction of the first predator. Consequently, the second predator relies on a food source that is independent of the first predator in order to ensure its own existence.

Existence and stability of a coexisting equilibrium point represents the optimal outcome for a biosystem. 
Ecologically, a coexisting equilibrium point depicts a scenario in which all of the species in a biosystem may coexist indefinitely. In our model, we have determined analytically and numerically the existence and stability of the coexisting equilibrium point. It establishes the fact that in our system, all three species can coexist under certain conditions.
If the first predator's growth rate is lower than its death rate, the first predator cannot survive because it lacks an alternate food source. Coexistence is therefore impossible.  

\par The mortality rate of a species is critical to sustaining coexistence in a biosystem. When the per capita death rate of one species continues to rise, the equilibrium density of that species declines while the equilibrium density of its competitor increases\cite{Abra}. Disease-related death, human-induced death, and other factors can all contribute to an increase in a predator's death rate. In this model, it can be seen that when the death rate of the second predator is less than its growth rate and the growth rate of the first predator is very low, the survival of the second predator depends entirely on its ability to hunt its prey, and kleptoparasitism doesn't help the second predator very much. However, if the second predator's death rate is higher than its growth rate, it will not be able to maintain its biomass through predation alone and will instead need to rely on kleptoparasitism, making it highly dependent on the first predator population. When the latter declines, the second predator's biomass will decline as well. Also, it is evident that at very high death rates for second predators, they cannot survive and hence vanish, as observed in this system.
The occurrence of saddle-node bifurcation for $A_7$ which can be seen in figure (\ref{A7}) provides the highest threshold value of $A_7$ under the parametric conditions set in table  (\ref{table:1}) beyond which coexistence is impossible. However, below this threshold, kleptoparasitism contributes to the maintenance of coexisting equilibrium. When the second predator's death rate is very low and growth rate is greater than death rate, then both predators' survival is dependent on their functional responses, and thus, due to an increase in interspecific competition between the predators, the first predator vanishes from this system, triggering a transcritical bifurcation for the parameter $A_7$ as shown in figure (\ref{A7}), changing the stability of the coexisting equilibrium point $E^*$, and as a result, $E^*$ becomes unstable and $E_3$ becomes stable.
When the death rate of the first predator i.e., $A_3$ is sufficiently high, the system undergoes a transcritical bifurcation, which changes the stability of the equilibrium points $E^*$ and it becomes unstable, and the axial equilibrium point $E_1$ becomes stable, which means ecologically that first predator will vanish, but alongwith it second predator will also extinct. However, when the first predator's death rate ($A_3$) is exceedingly low, a saddle-node bifurcation occurs in the system  as shown in figure (\ref{A3}), indicating the minimal threshold value of $A_3$ below which coexistence of all three species becomes unfeasible. This represents the scenario in which first predator overeats prey and as a consequence, the ecosystem either collapses or reaches axial equilibrium.

\par Hopf bifurcation does not occur around the equilibrium points $E_1,E_2,E_3$, implying that periodic oscillations in population dynamics do not occur when one or both predators become extinct.

\par Many occurrences of generalised Hopf bifurcation are found in this system which can be seen in  as shown in figure (\ref{GH}). Using the generalised Hopf bifurcation, the second Lyapunov coefficient can be used to calculate the relative position of the stable and unstable limit cycles.

\par A species' growth rate is also a significant aspect in determining its long-term survival. In our biosystem, $A_4$ is related to the first predator's growth rate. The  figure (\ref{A4}) shows that for a very low value of $A_4$, if the mortality rate of the second predator is greater than its growth rate, the system undergoes a transcritical bifurcation in which coexisting equilibrium loses its stability and axial equilibrium becomes stable. At $A_4=1.2252$, the system goes through a subcritical Hopf bifurcation, in which a coexisting equilibrium changes its stability and goes from stable to unstable  which can be seen in figures (\ref{A4}),(\ref{before a4}),(\ref{after A4}).  Subcritical Hopf bifurcation can be identified by the limit cycle that results from the bifurcation that is unstable and overlaps the steady state in parameter space. In this scenario, the three species populations reach the stable coexisting equilibrium point with a damped periodic oscillation, which is conceivable due to the attraction of the attractive equilibrium point and repelling limit cycle. From figures (\ref{A4}),(\ref{before a4s}),(\ref{after a4s}), it can be seen that at $A_4=12.071019$, a supercritical Hopf bifurcation occurs, changing the stability of the coexisting equilibrium point. As a result, the coexisting equilibrium point becomes stable until a saddle-node bifurcation occurs, giving the threshold maximum value of the parameter $A_4$, after which no coexisting equilibrium can be found. A reduction in the value of $A_4$ below the supercritical Hopf bifurcation point destabilises the steady coexisting equilibrium point, and the initial coexisting population begins to fluctuate periodically due to the birth of a stable limit cycle, and as the parameter decreases further, the diameter of the stable limit cycle rises, and due to the instability of the coexisting equilibrium point, any coextant population nearby to the equilibrium point begins to fluctuate with a variable periodicity and, in the long run, attains a stable periodicity and continue to fluctuate in that fixed periodicity as the trajectory reaches the stable limit cycle. If any population starts its journey close to the stable limit cycle but outside of it, it will still reach the same level and reach a fluctuation with a definite period in the long run. Also, as the value of the parameter $A_4$ increases, the prey population continues to decline, the first predator gradually increases, and the second predator gradually diminishes until it vanishes from the system due to a transcritical bifurcation. The parameter $A_5$ is related to the second predator's growth rate from food gained via hunting prey.
As the value of $A_5$ increases, the population of prey and the second predator continues to slowly increase, but the population of the first predator continues to decrease, and after a period of time, it disappears from the system as a result of a transcritical bifurcation as shown in figures (\ref{A5}). A supercritical Hopf bifurcation occurs at $A_5=2.007197$, and the value of the first Lyapunov coefficient becomes negative. It shifts the stability of the coexisting equilibrium point to produce a stable periodic oscillation in the population of all three species in the long term which can be seen from figures (\ref{A5}),(\ref{before a5s}),(\ref{after a5s}) .

\par One of the important feeding strategies of a predator is kleptoparasitism. Kleptoparasitism plays a crucial role in sustaining coexistence in this system. Kleptoparasitism is critical to maintaining coexistence in this system. When the rate of kleptoparasitism(\emph{a}) by the second predator is high, it can be seen that the growth rate of the first predator goes to zero, which is what happens in nature. Due to the high kleptoparasitism rate of the second predator, practically every prey killed by the first predator is stolen by the second predator, leaving nothing for the first predator to devour, and since predation is the only way for the first predator to survive, the growth rate becomes zero. Consequently, the first predator species becomes extinct. Now, if the death rate exceeds the growth rate of the second predator, it perishes alongside the first predator, and the axial equilibrium point becomes stable. However, if the growth rate exceeds the death rate of the second predator, it can coexist with the prey species and the equilibrium point $E_3$ becomes stable. When the rate of kelptoparasitism is very low, i.e. when nearly no food of the first predator is stolen by the second predator, then $\lim_{a \to 0} A_8=0$, indicating that kleptoparasitism does not contribute much to the second predator's growth rate, and so the survival of both predator species is dependent on their functional responses. Thus, depending on the rates of interspecific competition, cohabitation between all three species is conceivable. These scenarios are highly typical from an ecological standpoint, proving that our model is very accurate in describing these specific ecological scenarios. Here, the parameter $A_8$ is related to the growth rate of the second predator due to the food obtained by kleptoparasitism. The appearance of saddle-node bifurcation at $A_8= 10.20$ yields the lower-limit of this parameter below which coexistence is impossible as shown in figures (\ref{A8}). In a nutshell, it is clear from the dynamical scenario of the model (\ref{c}) discussed above that kleptoparasitism growth-related parameter $A_8$ is crucial in determining the saddle-node and transcritical bifurcations of the coexisting equilibrium point, which in turn affects the stability of the interior equilibrium.

\end{document}